\documentclass[11pt, reqno]{amsproc}
\usepackage{epstopdf}
\usepackage{mathrsfs, color}
\usepackage{stmaryrd}
\usepackage{cases}
\usepackage{amssymb}
\usepackage{amsmath}
\usepackage{amsfonts}
\usepackage{amsmath,amstext,amsbsy,amssymb}
\usepackage{verbatim}
\usepackage{float}
\usepackage{graphicx}

\newtheorem{theorem}{Theorem}[section]

\theoremstyle{definition}

\theoremstyle{remark}

\numberwithin{equation}{section} \errorcontextlines=0

\begin{document}
{\title[Analytic Formulas for Quantum Discord of Special Families of N-Qubit States]
{Analytic Formulas for Quantum Discord of Special Families of N-Qubit States}


%
%

\author{Jianming Zhou$^1$}
%
\author{Xiaoli Hu$^{2,*}$}
%
\author{Honglian Zhang$^3$}
%
\author{Naihuan Jing$^4$}

\thanks{1 School of Mathematics and Physics, University of South China, Hengyang, Hunan 421001, China}
\thanks{2 School of Artificial Intelligence, Jianghan University, Wuhan, Hubei 430056, China}
\thanks{3 Department of Mathematics, Shanghai University, Shanghai 200444, China }
\thanks{4 Department of Mathematics,
   North Carolina State University,
   Raleigh, NC 27695, USA} 
\thanks{* Corresponding to xiaolihumath@jhun.edu.cn}    


\begin{abstract}
Quantum discord, a key indicator of non-classical correlations in bipartite systems, has been recently extended to multipartite scenarios [Phys. Rev. Lett. 2020, 124:110401]. We present exact analytic formulas for the quantum discord of special families of N-qubit states, including
generalized class of GHZ states. 
Our formulations span $2$, $3$, $4n$, $4n+1$, $4n+2$, and $4n+3$-qubit configurations where $n\in 1, 2, \ldots$, which
refine the assessment of quantum correlations and provide an analytical tool in quantum computation. Moreover, we uncover a
``discord freezing'' in even-qubit systems under phase flip decoherence which provides a means for preserving quantum coherence in environmental perturbations.

\end{abstract}
\maketitle
\section{Introduction}
As the quantum era takes shape with the advent of large-scale quantum networks, including quantum computing and quantum communications, the need to harness non-classical correlations for quantum information processing has become increasingly critical. Among quantum correlations, multipartite quantum discord has gained attention due to its fundamental role in quantum mechanics and potential applications in enhancing quantum technology capabilities \cite{ XSF,ABM,MCC,BCP,SSDB,SAP,SAVZ,XY,XKH,H}.

Quantum discord, introduced by Ollivier and Zurek \cite{OZ} and Henderson and Vedral \cite{HV}, is a measure of quantum correlation more general than entanglement. It quantifies the nonclassicality of correlations as the discrepancy in quantum mutual information due to measurement disturbance on a subsystem \cite{RS,LF}. Computation of multipartite quantum discord presents a formidable challenge. It is commonly known that the computation is NP-complete due to the exponential increase in variables with the growing number of qubits \cite{H}, and analytical solutions have been confined to selected classes of quantum states \cite{Al,VR,XZCY,L,LC,ARA,ZHJ,ZHJ12,ZHJ13,JY,HXZ,ZHL,LZL}.

For a bipartite state $\rho$ of system $A\otimes B$, quantum discord is originally defined \cite{OZ,HV} through a minimization process:
\begin{equation}
\begin{split}
\mathcal{Q}(\rho) = \min_{\Pi^A}{[S_{B|\Pi^A}(\rho)-S_{B|A}(\rho)]},
\end{split}
\end{equation}
where $S_{B|\Pi^{A}}(\rho)= \sum_{j}p_jS(\rho_j)$ is the measured conditional entropy following a projective measurement $\Pi^A_j$ on system $A$ with $\rho_j = \Pi_j^A\rho\Pi_j^A/p_j$, $p_j = \mathrm{tr}(\Pi_j^{A}\rho\Pi_j^{A})$, and $S_{B|A}(\rho) = S(\rho)-S(\rho^{A})$ is the unmeasured conditional entropy of $B$ given $A$ \cite{OZ,CA16,CA17}. Here $S(\rho^{X}) = -\mathrm{tr}[\rho^{X}\log_2(\rho^{X})]$ represents the von Neumann entropy of system $X$'s reduced density matrix and in this context we define the von Neumann entropy of all larger density matrices as $S(\rho)\equiv-\mathrm{tr}[\rho\log_2(\rho)]$.

Recent advancements have generalized the quantum discord framework to $N$-partite systems by incorporating conditional dependencies that arise during successive measurement processes \cite{RLB}. Specifically, for an $N$-partite quantum state $\rho$, the measurements are stipulated to be sequential, with each measurement contingent upon the outcomes of its predecessors, forming a chain of conditional operations.
The $N-1$ partite measurement can be mathematically defined by:
\begin{equation}
\begin{split}
\Pi_{j_1\cdots j_{N-1}}^{A_1\cdots A_{N-1}} = \Pi_{j_1}^{A_1}\otimes \Pi_{j_2|j_1}^{A_2}\otimes\cdots\otimes \Pi^{A_{N-1}}_{j_{N-1}|j_1\cdots j_{N-2}},
\end{split}
\end{equation}
where ${\Pi_{j_k|j_1,\ldots,j_{k-1}}^{A_k}}$ ($j_1, \ldots, j_{k-1}=0, 1$) is a von Neumann projection operator of subsystem $A_k$ onto a one-dimensional space that is conditional on the measurement outcomes of the subsystems $A_1,\cdots, A_{k-1}$.
As an example, $\Pi_{j_2|j_1}^{A_2}$ represents a projector onto system $A_2$ conditioned on the measurement outcome of $A_1$, such that $\sum_{j_2}\Pi^{A_2}_{j_2|j_1} = \mathbb{I}^{A_2}$ and $\sum_{j_1}\Pi^{A_1}_{j_1} = \mathbb{I}^{A_1}$. This methodology preserves a sequential measurement protocol that progresses in an orderly manner from  $A_1$ through to  $A_N$.

Subsequently, upon quantifying the conditional entropies informed by these structured measurements, one defines the quantum discord of an $N$-partite state $\rho$ as:
\begin{equation}\label{1.2}
	\begin{split}
		\mathcal{Q}_{A_1,A_2,\ldots,A_N}(\rho) =&\min_{\Pi^{A_1\ldots A_{N-1}}}\big[S_{A_2|\Pi^{A_1}}(\rho)+S_{A_3|\Pi^{A_1A_2}}(\rho)\\
		+&\ldots+S_{A_N|\Pi^{A_1\ldots A_{N-1}}}-S_{A_2\ldots A_N|A_1}(\rho)\big].
	\end{split}
\end{equation}
where $\Pi^{A_1\ldots A_{k}}$ is an abbreviation of $\Pi_{j_1\cdots j_{k}}^{A_1\cdots A_{k}}$ for $k=1,\cdots,N-1$, the conditional entropy $S_{A_k|\Pi^{A_1,\ldots,A_{k-1}}} (\rho) $ corresponds to the post-measurement entropy of subsystem $A_k$, and is explicitly given through the ensemble $\{p_{j_1,\ldots,j_{k-1}} , \rho_{j_1,\ldots,j_{k-1}}^{A_1,\ldots,A_{k-1}}\}$ obtained post-measurement.
The ensemble elements are specified as:
\begin{equation}
\begin{split}
\rho_{j_1,\ldots,j_{k-1}}^{A_1,\ldots,A_{k-1}} =& \frac{1}{p_{j_1,\ldots,j_{k-1}}}\Pi_{j_1\cdots j_{k-1}}^{A_1\cdots A_{k-1}}~\rho~\Pi_{j_1\cdots j_{k-1}}^{A_1\cdots A_{k-1}},
\end{split}
\end{equation}
with the probabilities
\begin{equation}\begin{split}
	p_{j_1,\ldots,j_{k-1}} = &\text{tr}\big(\Pi_{j_1\cdots j_{k-1}}^{A_1\cdots A_{k-1}} ~\rho~\Pi_{j_1\cdots j_{k-1}}^{A_1\cdots A_{k-1}}\big),
\end{split}\end{equation}
which quantify the likelihood of measurement outcomes. Here, $S_{A_2,\ldots,A_N|A_1}(\rho)$  $= S(\rho) - S(\rho^{A_1})$ is the unmeasured conditional entropy for the composite system without $A_1$. 

The complexity of multipartite quantum discord computation is profound, especially due to the involved minimization process. In this study, we explore the characteristics and calculation of quantum discord within a multi-qubit state family, and deliver an analytical formula for the quantum discord of the generalized multi-partite Greenberger-Horne-Zeilinger (GHZ) state.

\section{Quantum Discord of non-symmetric Multi-Qubit state}
In our analysis, we focus on a special class of multipartite quantum states defined as follows:
\begin{equation}\label{e:2.1}
\begin{split}
\rho=&\frac{1}{2^N}\Big(\mathbb{I}+\sum_{i=1}^3c_i\sigma_i\otimes\cdots\otimes\sigma_i+s\sigma_3\otimes \mathbb{I}\otimes\cdots\otimes \mathbb{I}\\
&+s \mathbb{I}\otimes\sigma_3 \otimes\cdots\otimes \mathbb{I}+\cdots + s\mathbb{I}\otimes \mathbb{I}\otimes\cdots\otimes \sigma_3\Big),
\end{split}
\end{equation}
where $\sigma_1,\sigma_2,\sigma_3$ represent the Pauli spin matrices and $\mathbb{I}$ denotes the identity matrix. For the simplest case of $N=2$, this family includes the Bell states, which are notable for historical and foundational importance in quantum information science. Furthermore, this case also corresponds to the general $X$-type quantum states, for which Jing et. al. \cite{JY} has provided explicit expressions for the quantum discord.

These findings have been the subject of substantial interest within the field of quantum information due to their implications for understanding correlations in multipartite systems. In particular, the extension to arbitrary $N$ covers a set of highly symmetric $N$-qubit states explored in \cite{ZHL}. 


We first consider the tripartite qubit state $\rho$ on the joint Hilbert space
$H_{A_1} \otimes H_{A_2} \otimes H_{A_3} =$ $\mathbb{C}^2 \otimes \mathbb{C}^2 \otimes \mathbb{C}^2$ defined by
\begin{equation}\label{e:2.2}
\begin{split}
\rho&=\frac{1}{8}\Big(\mathbb{I}+\sum_{i=1}^3c_i\sigma_i\otimes\sigma_i\otimes\sigma_i+s\sigma_3\otimes \mathbb{I}\otimes \mathbb{I}+s \mathbb{I}\otimes\sigma_3\otimes \mathbb{I}+ s\mathbb{I}\otimes \mathbb{I}\otimes \sigma_3\Big).
\end{split}
\end{equation}
This state is of particular interest as one can see how bipartite entanglement is extended to tripartite systems, and also serves as a fundamental example for exploring multipartite quantum correlations.

Define the following function for convenience
\begin{equation}\label{e:H}
\begin{aligned}
H_{y}(x)=(1+y+x)\log_2(1+y+x)
+(1+y-x)\log_2(1+y-x)
\end{aligned}
\end{equation}
and abbreviated as $H(x)$ when $y=0$.
Upon performing detailed computations, we have derived analytical expressions for the quantum discord for the multipartite state as follows (proof is available in Appendix).
\begin{theorem}\label{e:Thm1}
Let $c=\{|c_1|,|c_2|\}$. Then the quantum discord $\mathcal{Q}_{A_1,A_2,A_3}(\rho)$ of the three-qubit state $\rho$ as specified in \eqref{e:2.2} can be analytically given by the following cases:

Case 1 : if $s^2c_3\leq (1-|s|)(c_3^2-c^2)$ or $c_3\leq 0$ and $c_3^2\geq c^2$, 
\begin{equation}
\begin{split}
\mathcal{Q}_{A_1,A_2,A_3}(\rho)=&\sum_i\lambda_i\log_2\lambda_i+3-\frac{1}{8}[H_{2s}(|s+c_3|)+H_{-2s}(|s-c_3|)\\
&+H(|s+c_3|)+H(|s-c_3|)],
\end{split}
\end{equation}
where $(\text{for}~~k=1,\ldots,6, l=7,8)$
\begin{equation}
\begin{split}\lambda_{k} &= \frac{1}{8}\Big(1+(-1)^k\sqrt{c_1^2+c_2^2+c_3^2-2c_3s+s^2}\Big); \\
\lambda_{l} &= \frac{1}{8}\Big(1+(-1)^l\sqrt{c_1^2+c_2^2+c_3^2+6c_3s+9s^2}\Big) \end{split}
\end{equation}
are the eigenvalues of $\rho$ in \eqref{e:2.2}.

In the special case where $|c_1|=|c_2|=|c_3|=c\in (-1,0]$, we have
\begin{equation}
\begin{split}
\mathcal{Q}_{A_1,A_2,A_3}(\rho)
=&\sum_i\lambda_i\log_2\lambda_i+3-\frac{1}{8}[H_{2s}(|s+c|)+H_{-2s}(|s-c|)\\
&+H(|s+c|)+H(|s-c|)].
\end{split}
\end{equation}

Case 2 : if $s=0$, define $C=\max{\{|c_1|,|c_2|,|c_3|\}}$, then
\begin{equation}\label{thm2case2}
\begin{split}
\mathcal{Q}_{A_1,A_2,A_3}(\rho)=\frac{1}{2}H(\sqrt{c_1^2+c_2^2+c_3^2})-\frac{1}{2}H(C),
\end{split}
\end{equation}
\end{theorem}

Next we investigate the following four-qubit system, delineated by the subsystems $A_1$, $A_2$, $A_3$ and $A_4$:
\begin{equation}\label{e:fourdef}
\begin{split}
\rho=&\frac{1}{16}(\mathbb{I}+\sum_{i=1}^3c_i\sigma_i\otimes\sigma_i\otimes\sigma_i\otimes\sigma_i+s\sigma_3\otimes \mathbb{I}\otimes \mathbb{I}\otimes \mathbb{I}+s \mathbb{I}\otimes\sigma_3\otimes \mathbb{I}\otimes \mathbb{I}\\
&+ s\mathbb{I}\otimes \mathbb{I}\otimes \mathbb{I}\otimes \sigma_3).
\end{split}
\end{equation}
We have devised and computed the analytical formulas of the quantum discord given in the following theorem.
\begin{theorem}\label{e:Thm2}
 For the family of four-qubit states represented by \eqref{e:fourdef}, the computation of quantum discord splits into distinct cases, which we have explicated herein.

Case 1 : if $c_3\leq 0$ and $c_3^2\geq c^2$ or $\frac{s^2}{(1-2|s|)}\geq \frac{(c_3^2-c^2)}{c_3}$, then the quantum discord is given by
\begin{equation}
\begin{split}
\mathcal{Q}_{A_1,A_2,A_3,A_4}(\rho)
=&\sum_i\lambda_i\log_2\lambda_i+4-\frac{1}{16}\big[H_{3s}(|s+c_3|)+H_{-3s}(|s-c_3|)\\
&+3H_{s}(|s-c_3|)
+3H_{-s}(|s+c_3|)\big],
\end{split}
\end{equation}
where $(\text{for}~~j=1,\ldots,6,\, k=7,\ldots,14,\, l=15,16)$
\begin{equation}
\begin{split}\lambda_{j} &= \frac{1}{16}\Big(1+(-1)^j(c_1+c_2+c_3)\Big);\\
\lambda_{k} &= \frac{1}{16}\Big(1-c_3+(-1)^k\sqrt{c_1^2+c_2^2-2c_1c_2+4s^2}\Big);\\
\lambda_{l} &= \frac{1}{16}\Big(1+c_3+(-1)^l\sqrt{c_1^2+c_2^2+2c_1c_2+16s^2}\Big)
\end{split}
\end{equation}
are the eigenvalues of $\rho$ in \eqref{e:fourdef}.

Case 2 : if $s=0$, define $C=\max{\{|c_1|,|c_2|,|c_3|\}}$, then
\begin{equation}\label{e:thm2case2}
\begin{split}
\mathcal{Q}_{A_1,A_2,A_3,A_4}(\rho)=&\frac{1}{4}\Big[(1+c_1+c_2+c_3)\log_2(1+c_1+c_2+c_3)\\
&+(1+c_1-c_2-c_3)\log_2(1+c_1-c_2-c_3)\\
&+(1-c_1+c_2-c_3)\log_2(1-c_1+c_2-c_3)\\
&+(1-c_1-c_2+c_3)\log_2(1-c_1-c_2+c_3)
\Big]\\
&-\frac{1}{2}H(C).
\end{split}\end{equation}
\end{theorem}
Now we are ready to study $N$-partite quantum system.
Our rigorous derivation culminates in the following generalization
\begin{theorem}\label{e:ThmN}
For the class of states within $N$-qubit systems as stipulated by (\ref{e:2.1}), we have the following explicit formulas for the quantum discord.

Case 1: If $c_3\leq 0$ , $c^2\leq c_3^2$ or $\frac{s^2}{1-(N-2)|s|}\geq \frac{c_3^2-c^2}{c_3}$, then the quantum discord is quantified as:
\begin{equation}
\begin{split}
\mathcal{Q}_{A_1,A_2,\ldots,A_N}(\rho)&=\sum_i\lambda_i\log_2\lambda_i+N-\max{W},
\end{split}
\end{equation}
 where $\lambda_i$ are the eigenvalues associated with $\rho$, and the term $\max W$ is determined in accordance with the subsequent details.\\
(i) if $N=4n$, $n\in\mathbb{N}^+$, then
\begin{equation}\label{e:4n}
\begin{split}
\max{W} =& \frac{1}{2^{4n}}\sum_{k=0}^{2n-1}\tbinom{4n-1}{2k}H_{(4n-4k-1)s}(|s+c_3|)\\&+\frac{1}{2^{4n}}\sum_{k=0}^{2n-1}\tbinom{4n-1}{2k+1}H_{(4n-4k-3)s}(|s-c_3|);
\end{split}
\end{equation}
(ii) if $N=4n+1$, $n\in\mathbb{N}^+$, then
\begin{equation}\label{e:4n+1}
\begin{split}
\max{W} = & \frac{1}{2^{4n+1}}\sum_{k=0}^{2n}\tbinom{4n}{2k}H_{(4n-4k)s}(|s+c_3|)\\&+\frac{1}{2^{4n+1}}\sum_{k=0}^{2n-1}\tbinom{4n}{2k+1}H_{(4n-4k-2)s}(|s-c_3|);
\end{split}
\end{equation}
(iii) if $N=4n+2$, $n\in\mathbb{N}^+$, then
\begin{equation}\label{e:4n+2}
\begin{split}
\max{W} = & \frac{1}{2^{4n+2}}\sum_{k=0}^{2n}\tbinom{4n+1}{2k}H_{(4n-4k+1)s}(|s+c_3|)\\&+\frac{1}{2^{4n+2}}\sum_{k=0}^{2n}\tbinom{4n+1}{2k+1}H_{(4n-4k-1)s}(|s-c_3|);
\end{split}
\end{equation}
(iv) if $N=4n+3$, $n\in\mathbb{N}^+$, then
\begin{equation}\label{e:4n+3}
\begin{split}
\max{W} = & \frac{1}{2^{4n+3}}\sum_{k=0}^{2n+1}\tbinom{4n+2}{2k}H_{(4n-4k+2)s}(|s+c_3|)\\&+\frac{1}{2^{4n+3}}\sum_{k=0}^{2n}\tbinom{4n+2}{2k+1}H_{(4n-4k)s}(|s-c_3|),
\end{split}
\end{equation}
where the binomial coefficient $\tbinom{N-1}{j}$ represents the number of combinations to choose $j$ elements out of $N-1$ elements.

Case 2: if $s=0$, define $C=\max{\{|c_1|,|c_2|,|c_3|\}}$, then the quantum discord of $\rho $ is given by
\begin{equation}
\begin{split}
\mathcal{Q}_{A_1,A_2,\ldots,A_N}(\rho)&=\sum_i\lambda_i\log_2\lambda_i+N-\frac{1}{2}H(C).
\end{split}
\end{equation}
\end{theorem}
The proof of these results is included in the Appendix.
In the following we list several important special cases, which underscore the intricate nature of quantum correlations present in such systems.
\begin{theorem}\label{e:ThmN2}
For the following family of $N$-qubit states,
\begin{equation}
\begin{split}
\rho =& \frac{1}{2^N}(\mathbb{I}+s_1\sigma_3\otimes \mathbb{I}\otimes\cdots\otimes \mathbb{I}+s_2\mathbb{I}\otimes \sigma_3\otimes\mathbb{I}\otimes\cdots\otimes \mathbb{I}\\&+\cdots+s_N\mathbb{I}\otimes\cdots\otimes\mathbb{I}\otimes\sigma_3),
\end{split}
\end{equation}
the quantum discord is given by
\begin{equation}
\begin{split}
\mathcal{Q}_{A_1,A_2,\ldots,A_N}(\rho) =& \sum_i\lambda_i\log_2\lambda_i+N
-\frac{1}{2^{N}}\sum_{u_i=0,1}H_{\sum_{i=1}^{N-1}(-1)^{u_i}s_i}(|s_{N}|).
\end{split}
\end{equation}
\end{theorem}

Next, we consider the important Greenberger-Horne-Zeilinger (GHZ) state. 
The generalized GHZ state is defined as: $|GHZ\rangle=\frac{1}{\sqrt{2}}(|0\cdots 0\rangle+|1\cdots 1\rangle)$.
This state embodies a deeply entangled N-qubit state with significant quantum correlations among its constituents. Here,
 we will study  quantum discord for a generalized GHZ state with certain degree of mixedness (determined by the parameter $\mu$), namely:
\begin{equation}
\rho_{GHZ} = \mu|GHZ\rangle\langle GHZ| + \frac{1-\mu}{2^N}\mathbb{I},
\end{equation}
In terms of the Pauli basis, the GHZ state can be written as follows:
\begin{equation}
\begin{split}\label{e:GHZPauliB}
\rho_{GHZ}=&\frac{1}{2^N}\mathbb{I}+\frac{\mu}{2^N}\sum_{t=1}^{\lfloor \frac{N}{2}\rfloor}\bigg[\sum_\pi \mathbf{S}_{\pi}\Big(\mathbb{I}^{\otimes (N-2t)}\otimes\sigma_3^{\otimes 2t}\Big)\\
&+(-1)^t\sum_\pi \mathbf{S}_{\pi}\Big(\sigma_1^{\otimes (N-2t)}\otimes\sigma_2^{\otimes 2t}\Big)\bigg]+\frac{\mu}{2^N}\sigma_1^{\otimes N},
\end{split}
\end{equation}
where $\lfloor ~ \rfloor $ denotes the floor function, and the symmetric group $\mathbf{S}_{\mathbf{\pi}}$ acts on the tensor product by the permutation $\mathbf{\pi}\equiv \begin{pmatrix}
	1&2&\cdots N\\
	\pi_1&\pi_2&\cdots\pi_N
\end{pmatrix}$
 with $\pi_m\in \{1,\ldots,N\}$,
i.e, the permutation $\mathbf{\pi}$ sends $v_1\otimes v_2\cdots\otimes v_N$ to $v_{\pi_1}\otimes v_{\pi_2}\cdots\otimes v_{\pi_N}$. $\sum_{\pi}$ is the sum of all distinct permutations of the qubits.

\begin{figure}
\includegraphics[width=1\linewidth]{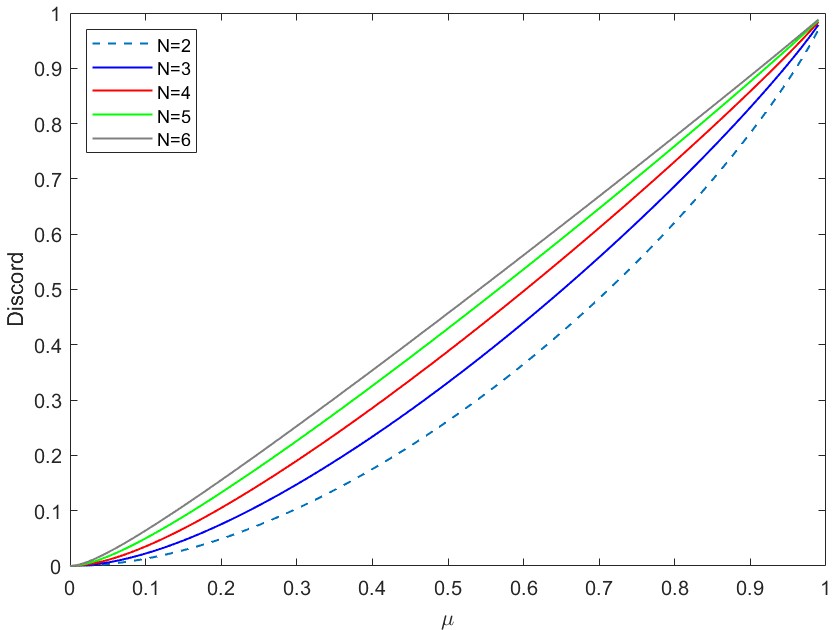}
\caption{(color online)\label{fig1} The graph of quantum discord for GHZ states with the parameter $\mu$ of system sizes from $N=2$ to $N=6$. It is observed that as $N$ increases, the curve closely approximates a linear function. This shows scaling properties of quantum discord in multipartite entangled systems, and may provide insight on quantum coherence and entanglement across dimension of the quantum system.}
\end{figure}

\begin{theorem} The quantum discord for the multipartite GHZ state  \eqref{e:GHZPauliB} can be explicitly computed by:
\begin{equation}
\begin{split}\label{e:QGHZ}
\mathcal{Q}_{A_1,\ldots,A_N}(\rho_{GHZ})
=&\frac{1-\mu}{2^N}\log_2(1-\mu)+\frac{1+(2^N-1)\mu}{2^N}\log_2(1+(2^N-1)\mu)\\
&-\frac{1+(2^{N-1}-1)\mu}{2^{N-1}}\log_2\Big(1+(2^{N-1}-1)\mu\Big).
\end{split}
\end{equation}
\end{theorem}

\begin{figure}[htbp]
\centering
\includegraphics[width=0.8\linewidth]{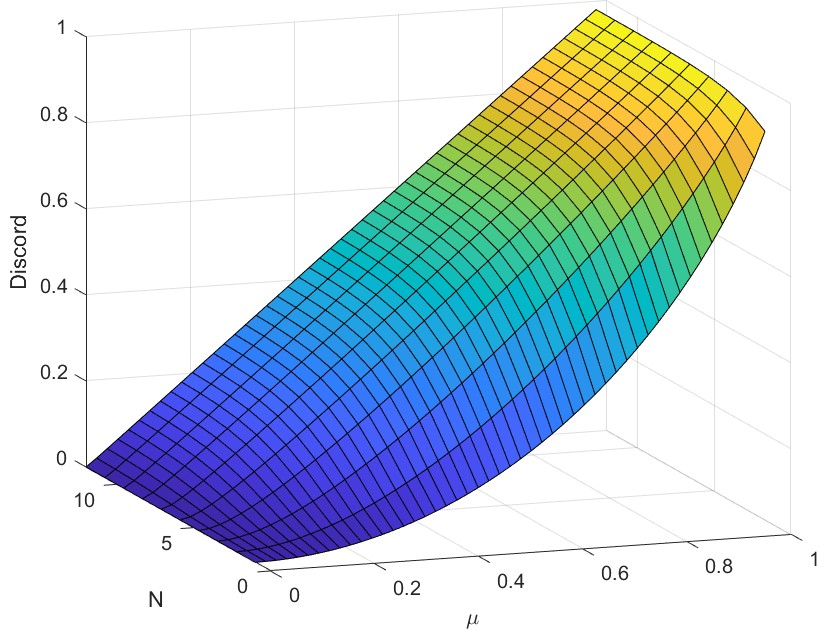}
\caption{\label{fig2}The behavior of quantum discord for $N$-partite GHZ state. Here, $N$ and $\mu$ have been gridded for better display. In fact, $N$ should be an integer.}
\end{figure}

The behavior of quantum discord for the $N$-partite GHZ state is depicted by Fig.\ref{fig1}.
We note that as the value of $N$ ascends, the curve of the quantum discord
progressively approximates the diagonal. We have plotted the striking trend in Fig.\ref{fig2}.

To analyze the behavior of quantum discord in multipartite GHZ states, we introduce a grid method to describe the system size ($N$) and the parameter ($\mu$). This allows us to present the results clearer, showing how quantum discord varies with different numbers of parties and control parameters. We highlight the importance that $N$ must be an integer because it stands for the count of parties in the GHZ state. Our gridded approach allows us to smoothly map out quantum discord's behavior as $N$ changes, while still recognizing the true discrete character of these quantum systems.
This study provides a clearer picture of how quantum discord behaves with different sizes of GHZ states, which helps broader understanding of quantum correlations in quantum computing and communication.

\section{DYNAMICS OF QUANTUM DISCORD UNDER LOCAL NONDISSIPATIVE CHANNELS}

Previous studies have reported that, for two-qubit states, quantum discord exhibits a remarkable stability of remaining invariant under certain decoherence channels within finite time intervals \cite{ZHJ,MPM,MCSV,HJM,BCA,LWF}. Moreover, a similar ``freezing'' behavior of quantum discord has been observed for even-qubit systems subjected to phase flip noise, confined to specific time periods \cite{LZL,H2023}. Extending this to the general $N$-qubit domain, we explore this phenomenon within the context of the GHZ states undergoing evolution via decoherence channel $\mathcal{E}$, characterized by its Kraus operators (cf.\eqref{Eq:Kraus} below). We hope to determine the persistence of quantum discord constancy across broader multi-partite settings, with implications for fault-tolerant quantum computing and information preservation.

The Kraus operators on the tensor space are defined by
\begin{equation}\label{Eq:Kraus}
\begin{split}
K_0^{(A_1)} = \Gamma_0\otimes \mathbb{I}\otimes \cdots \otimes \mathbb{I}&, \ \
K_1^{(A_1)} = \Gamma_1\otimes \mathbb{I}\otimes \cdots \otimes \mathbb{I},\\
& \vdots \\
K_0^{(A_N)} = \mathbb{I}\otimes \cdots \otimes \mathbb{I}\otimes \Gamma_0&,\ \
K_1^{(A_N)} = \mathbb{I}\otimes \cdots \otimes \mathbb{I}\otimes \Gamma_1,
\end{split}
\end{equation}
where
$$\Gamma_0 =\begin{pmatrix}
  \sqrt{1-p/2}&0\\
0&\sqrt{1-p/2}
\end{pmatrix}, ~~\Gamma_1 =\begin{pmatrix}
\sqrt{p/2}&0\\
0&-\sqrt{p/2}\end{pmatrix}$$ are the phase damping channels, and $p = 1-e^{-\gamma t}$ represents the decoherence probability, directly proportional to the elapsed time $t$ and with $\gamma$ denoting the characteristic phase damping rate.

For the three-qubit case ($N=3$), consider the state $\rho$ as defined in \eqref{e:2.2} subject to a phase flip channel. The evolved state is given by
\begin{equation}
\begin{split}
\mathcal{E}(\rho)
&=\sum_{\mu_i=0,1}K_{\mu_1}^{(A_1)}K_{\mu_2}^{(A_2)}K_{\mu_3}^{(A_3)}\rho [K_{\mu_3}^{(A_3)}]^{\dagger}[K_{\mu_2}^{(A_2)}]^{\dagger}[K_{\mu_1}^{(A_1)}]^{\dagger}\\
&=\frac{1}{8}\Big[\mathbb{I}+s\sigma_3\otimes \mathbb{I}\otimes \mathbb{I}+s \mathbb{I}\otimes\sigma_3 \otimes \mathbb{I}+ s\mathbb{I}\otimes \mathbb{I}\otimes \sigma_3+c_3\sigma_3\otimes\sigma_3\otimes\sigma_3\\
&+(1-p)^3c_2\sigma_2\otimes\sigma_2\otimes\sigma_2+(1-p)^3c_1\sigma_1\otimes\sigma_1\otimes\sigma_1\Big].
\end{split}
\end{equation}

Invoking Theorem \ref{e:Thm1}, the three-partite quantum discord $\mathcal{Q}_{A_1,A_2,A_3}$ for the state $\mathcal{E}(\rho)$ under the condition $|c_3| \geq \max\{|c_1|(1-p)^3,|c_2|(1-p)^3\}$ can be written as
\begin{equation}
	\begin{split}		\mathcal{Q}_{A_1,A_2,A_3}(\mathcal{E}(\rho))&=\frac{3}{8}H(\zeta)+\frac{1}{8}H(\eta)-\frac{1}{8}H_{2s}(|s+c_3|)-\frac{1}{8}H_{-2s}(|s-c_3|)\\
&-\frac{1}{8}H(|s+c_3|)-\frac{1}{8}H(|s-c_3|),
	\end{split}
\end{equation}
where the entropic quantities $\zeta$ and $\eta$ are delineated by
\begin{equation}
\begin{split}\zeta&= \sqrt{(c_1^2+c_2^2)(1-p)^6+(c_3-s)^2},\\
\eta&= \sqrt{(c_1^2+c_2^2)(1-p)^6+(c_3+3s)^2}.
\end{split}
\end{equation}
And for the specific case where $C = \max{\{|(1-p)^3c_1|, |(1-p)^3c_2|, |c_3|\}}$ and $s = 0$, the discord measure simplifies to
\begin{equation}
\begin{split}
\mathcal{Q}_{A_1,A_2,A_3}(\mathcal{E}(\rho))&=\frac{3}{8}H(\zeta)+\frac{1}{8}H(\eta)-\frac{1}{2}H(C).
\end{split}
\end{equation}

Under the phase flip channel, the quantum discord of the state $\mathcal{E}(\rho)$ undergoes a monotonically decreasing pattern with respect to the decoherence parameter $p$, i.e., $\frac{\partial \mathcal{Q}_{A_1,A_2,A_3}}{\partial p} < 0$. This deduction can be substantiated as follows.

If $|c_3|\geq \max{\{|c_1|(1-p)^3, |c_2|(1-p)^3\}}$, one can observe that
$H'(x) = \log_2\frac{1-x}{1+x}$ and $\zeta, \eta$ are nonnegative,
thus $\log_2\frac{1-\zeta}{1+\zeta}<0$.
Meanwhile, $\frac{\partial\zeta}{\partial p}<0$ as $0<p<1$. A similar result holds with respect to $\eta$. Therefore,
$\mathcal{Q}'_{A_1,A_2,A_3}(\mathcal{E}(\rho))<0.$
For $s = 0$, it is reduced to the results in \cite{LZL} and we can also obtain that $\mathcal{Q}'_{A_1,A_2,A_3}(\mathcal{E}(\rho))<0.$


When $s = 0$, the results corroborate previous findings reported in \cite{LZL}, further confirming that $\frac{\partial \mathcal{Q}_{A_1,A_2,A_3}}{\partial p} < 0$.

For the case of the four-qubit state $\rho$ given as \eqref{e:fourdef}, the action through a phase flip channel yields
\begin{equation}
\begin{split}
\mathcal{E}(\rho)=&\frac{1}{16}\Big(\mathbb{I}+(1-p)^4c_1\sigma_1\otimes\sigma_1\otimes\sigma_1\otimes\sigma_1+(1-p)^4c_2\sigma_2\otimes\sigma_2\otimes\sigma_2\otimes\sigma_2\\
&+s\sigma_3\otimes \mathbb{I}\otimes \mathbb{I}\otimes \mathbb{I}+c_3\sigma_3\otimes\sigma_3\otimes\sigma_3\otimes\sigma_3+s \mathbb{I}\otimes\sigma_3 \otimes \mathbb{I}\otimes \mathbb{I}\\
&+s\mathbb{I}\otimes \mathbb{I} \otimes \sigma_3\otimes \mathbb{I}+ s\mathbb{I}\otimes\mathbb{I}\otimes\mathbb{I}\otimes \sigma_3\Big).
\end{split}
\end{equation}
Utilizing Theorem \ref{e:Thm2}, the quantum discord $\mathcal{Q}(\mathcal{E}(\rho))$ under the condition $|c_3| \geq \max\{|c_1|(1-p)^4,|c_2|(1-p)^4\}$ is given by
\begin{equation}
\begin{split}
\mathcal{Q}_{A_1,A_2,A_3,A_4}(\mathcal{E}(\rho))
=&\frac{3}{16}H(e)+\frac{4}{16}H_{-c_3}(f)+\frac{1}{16}H_{c_3}(g)-\frac{1}{16}H_{3s}(|s+c_3|)\\
-&\frac{1}{16}H_{-3s}(|s-c_3|)-\frac{3}{16}H_{s}(|s-c_3|)-\frac{3}{16}H_{-s}(|s+c_3|),
\end{split}
\end{equation}
where the functions $e$, $f$, and $g$ are given by
\begin{equation}
	\begin{split}
		e &= (c_1+c_2)(1-p)^4+c_3;\\
		f &= \sqrt{(c_1-c_2)^2(1-p)^8+4s^2};\\
		g &= \sqrt{(c_1+c_2)^2(1-p)^8+16s^2},
	\end{split}
\end{equation}

In the special case when $s = 0$, and $C = \max\{|c_1(1-p)^4|,|c_2(1-p)^4|, |c_3|\}$, we derive a simplified expression for the quantum discord:
\begin{equation}\begin{split}
	\mathcal{Q}_{A_1,A_2,A_3}(\mathcal{E}(\rho)) &= \frac{3}{16} H(e) + \frac{1}{4} H_{-c_3}(f) + \frac{1}{16} H_{c_3}(g) - \frac{1}{2} H(C).\end{split}
\end{equation}

\begin{figure}[htbp]
\centering
\includegraphics[width=0.8\linewidth]{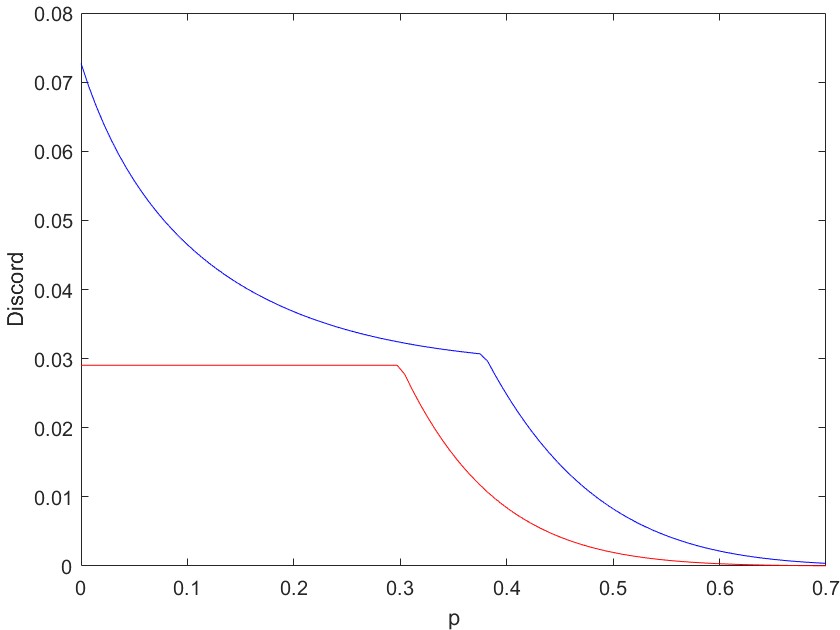}
\caption{\label{fig3}The graphical representations presented herein depict the comparative evolution of quantum discord as a function of the decoherence parameter, considering both three- and four-qubit quantum states. The former is exemplified by the blue curve, while the latter corresponds to the red curve. These plots are generated under the assumption of a specific decoherence channel characterized by parameters where the purity factor
$s=0$, the correlation coefficient $c_1=5/6$, and the coupling constants obey the relation $c_2=c_1c_3$
with $c_3=-0.2$. These conditions are deliberately chosen to illustrate the distinctive decoherence effects on quantum discord within multipartite systems.}
\end{figure}

Remarkably, if the system parameters satisfy $s = 0$, $c_2 = c_1c_3$, and $|c_1|(1-p)^4 \geq |c_3|$, the quantum discord becomes invariant over time:
\begin{equation}
	\mathcal{Q}_{A_1,A_2,A_3}(\mathcal{E}(\rho)) = \frac{1}{2} H(c_3),
\end{equation}
This invariance is attributed to temporal stability of the parameter $c_3$, aligning with results of \cite{LZL} and underlining the persistence of particular quantum correlations amidst decoherence.
For example, set $s=0$, $c_1 = 5/6$, $c_2=c_1c_3$ and $c_3=-0.2$.
The dynamic behavior of quantum discord under the decoherence channel is depicted in Figure \ref{fig3}. We find that for four-qubit, a sudden transition of quantum discord appears at $p=0.297$, but this phenomenon does not exist for three-qubit.
In other words, the quantum discord does not change under the decoherence channel during the time interval, which recovers the results in \cite{LZL}.
Moreover, this result extends to larger even-qubit systems, consolidating the correspondence between the discord of two-qubit, $4n-2$-qubit, and $4n$-qubit states at $s=0$, as established in \cite{LWF}.

The freezing phenomenon offers a promising pathway for preserving quantum coherence in quantum computing and information processing tasks. The ability to maintain stable quantum correlations in the face of decoherence could have significant implications for the development of quantum technologies, particularly in the context of error correction and the robustness of quantum communication protocols.

\section{Conclusions}
We have undertaken a comprehensive analytical computation of quantum discord in $N$-qubit systems, focusing on a generalized class of GHZ states. Our work 
 has given exact formulas for multipartite quantum discord in systems with qubits categorized into four distinct groups given by $4n$, $4n+1$, $4n+2$, and $4n+3$. This provides insight into the nature of multipartite quantum correlations, which may have applications in the realm of quantum information processing.

One aspect of the research is the examination of the dynamical properties of these
$N$-qubit states under the influence of phase flip decoherence. The observations reveal a remarkable phenomenon specific to even-qubit systems wherein quantum discord exhibits a ``freezing'' behavior, remaining invariant for specific finite durations. This resilience against decoherence contributes valuable knowledge to the stability of quantum states, which may have profound implications in the development of quantum technologies that rely on the preservation of quantum coherence.
This study highlights a form of resilience against decoherence, which is one of the major challenges in maintaining quantum coherence in practical quantum systems.
These findings have potential to influence a wide range of applications in quantum information processing, from the design of more stable quantum states to the development of new quantum communication protocols. Future research will likely explore the broader implications of these findings, particularly in the context of quantum state dynamics and the preservation of quantum coherence in noisy environments.

Overall, our results enrich the theoretical framework surrounding quantum discord, offering new possibilities for exploration in quantum communication protocols, and highlight possible future research in quantum state dynamics and preservation.
\vskip 0.1in
\centerline{\bf Acknowledgments}
\medskip
The research is supported in part by the NSFC
grants (nos.12171303, 12271332, 12126351 and 12226321), the Research Fund of Jianghan University (Grant No. 2023JCYJ08), the Natural Science Foundation of Shanghai (Grant No. 22ZR1424600) and the MP-TSM-00002518.
\vskip 0.1in

\textbf{Data Availability Statement.} All data generated during the study are included in the article.

\textbf{Conflict of Interest Statement.} The authors declare that the research was conducted in the
absence of any commercial or financial relationships that could be construed as a
potential conflict of interest.

\bibliographystyle{amsalpha}

\end{document}


{\title[Supplemental material:Analytic Formulas for Quantum Discord of Special Families of N-Qubit States]
{Supplemental material:Analytic Formulas for Quantum Discord of Special Families of $N$-Qubit States}
%
%
%

\author{Jianming Zhou$^1$}

\author{Xiaoli Hu$^2$*}

\author{Honglian Zhang$^3$}

\author{Naihuan Jing$^4$}

\thanks{1 School of Mathematics and Physics, University of South China, Hengyang, Hunan 421001, China} 
\thanks{2 School of Artificial Intelligence, Jianghan University, Wuhan, Hubei 430056, China}
\thanks{3 Department of Mathematics, Shanghai University, Shanghai 200444, China } 
\thanks{4 Department of Mathematics, North Carolina State University, Raleigh, NC 27695, USA}
\thanks{* Corresponding to xiaolihumath@jhun.edu.cn}    

\maketitle

\section*{Appendix A: Quantum Discord for a family of non-symmetric Multi-Qubit state}
\setcounter{equation}{0}
\setcounter{theorem}{0}
\renewcommand\theequation{A\arabic{equation}}
\renewcommand\thetheorem{2.\arabic{theorem}}

According to the definition of the generalized discord for multi-partite states in \cite{RLB},
the quantum discord of an $N$-partite state $\rho$ is given by
\begin{equation}\label{A1.2}
\begin{split}
\mathcal{Q}_{A_1,A_2,\ldots,A_N}(\rho)=&\min_{\Pi^{A_1\ldots A_{N-1}}}\big[S_{A_2|\Pi^{A_1}}(\rho)+\ldots+S_{A_N|\Pi^{A_1\ldots A_{N-1}}}\\
&-S_{A_2\ldots A_N|A_1}(\rho)\big],
\end{split}
\end{equation}
where $S_{A_2\ldots A_N|A_1}(\rho) = S(\rho)-S(\rho_{A_1})$ is the unmeasured conditional entropy with $S(\rho_{X}) = -\mathrm{tr}[\rho_{X}\log_2\rho_{X}]$ the von Neumann entropy of the quantum state on system X, and
$S_{A_k|\Pi^{A_1\cdots A_{k-1}}}(\rho) = \sum_{j_1\cdots j_{k-1}}p_j^{(k-1)}S_{A_1\ldots A_{k-1}}(\rho^{(k-1)}_j)$ with $\Pi_j^{(k-1)} = \Pi^{A_1\cdots A_{k-1}}_{j_1\cdots j_{k-1}}$, $\rho_j^{(k-1)} = 1/p^{(k-1)}_j\Pi_j^{(k-1)}\rho\Pi_j^{(k-1)}$ and $p_j^{(k-1)} = \mathrm{tr}(\Pi_j^{(k-1)}\rho\Pi_j^{(k-1)})$.

Considering the following $N$-partite states
\begin{equation}\label{A2.1}
\begin{split}
\rho&=\frac{1}{2^N}(\mathbb{I}+\sum_{i=1}^3c_i\sigma_i\otimes\cdots\otimes\sigma_i+s\sigma_3\otimes \mathbb{I}\otimes\cdots\otimes \mathbb{I}+s \mathbb{I}\otimes\sigma_3 \otimes\cdots\otimes \mathbb{I}\\
&+\cdots + s\mathbb{I}\otimes \mathbb{I}\otimes\cdots\otimes \sigma_3),
\end{split}
\end{equation}
where $\sigma_1,\sigma_2,\sigma_3$ represent the Pauli matrices, and $\mathbb{I}$ is the identity operator.

First, we consider the case of $N =3$, the state

\begin{equation}\label{A2.2}
\begin{split}
\rho=\frac{1}{8}(\mathbb{I}+\sum_{i=1}^3{c_i}\sigma_i\otimes\sigma_i\otimes\sigma_i+s\sigma_3\otimes \mathbb{I}\otimes \mathbb{I}
+s \mathbb{I}\otimes\sigma_3\otimes \mathbb{I}+ s\mathbb{I}\otimes \mathbb{I}\otimes \sigma_3).
\end{split}
\end{equation}

Recall the function: $H_{y}(x)=(1+y+x)\log_2(1+y+x)+(1+y-x)\log_2(1+y-x)$.
When $y=0$, $H_y(x)$ can be abbreviated as $H(x)$.
The reduced state on subsystem $A_1$ is $\rho_{A_1} = \mathrm{tr}_{A_2A_3}(\rho) = \frac{1}{2}(\mathbb{I}+s\sigma_3)$ and the corresponding entropy  is $S(\rho_{A_1}) = 1-\frac{1}{2}H(s)$. Therefore
\begin{equation}\label{A2.3}
\begin{split}
S_{A_2A_3|A_1}(\rho) = S(\rho)-S(\rho_{A_1}) = -\sum_i\lambda_i\log_2\lambda_i-1+\frac{1}{2}H(s),
\end{split}
\end{equation}
where $\lambda_i$ are eigenvalues of $\rho$.
For the computational basis $\{ |j\rangle,~j=0,1\}$ of $\mathbb{C}^2$, let $\Pi_j = |j\rangle\langle j|$, then the von Neumann measurement on subsystem $A_1$ is given by $\{P^{A_1}_{j} = V_{A_1}\Pi_jV_{A_1}^\dagger,~ j=0,1\}$ for some unitary matrix $V_{A_1}\in \mathrm{U}(2)$.
And $V_{A_1} = t_{A_1}\mathbb{I}+\sqrt{-1}\vec{y}_{A_1}\cdot\vec{\sigma}$ are the Bloch expansion of special unitary matrices
with real Bloch coefficients $t_{A_1} = \frac{1}{2}\mathrm{tr}(\mathbb{I}V_{A_1}), y_{{A_1},k} = \frac{1}{2}\mathrm{tr}[(\sqrt{-1}\sigma_k)^{\dagger}V_{A_1}]= \frac{-\sqrt{-1}}{2}\mathrm{tr}[\sigma_kV_{A_1}]$ for $k\in 1,2,3$ and $y_{{A_1}1}^2+y_{{A_1}2}^2+y_{{A_1}3}^2+t_{A_1}^2=1$.

After the measurement $\{P^{A_1}_{j}\}$, the state $\rho$ will become the ensemble $\{\rho_j,p_j\}$ with
$\rho_j = (P^{A_1}_{j}\otimes \mathbb{I})\rho(P^{A_1}_{j}\otimes \mathbb{I})/p_j$ and probability $p_j = \mathrm{tr}[(P^{A_1}_{j}\otimes \mathbb{I})\rho(P^{A_1}_{j}\otimes \mathbb{I})]$.
By use of the follow symmetric relations
{\small{\begin{equation}
\begin{split}
V_{A_1}^{\dagger}\sigma_1V_{A_1}&=(t_{A_1}^2+y_{{A_1}1}^2-y_{{A_1}2}^2-y_{{A_1}3}^2)\sigma_1+2(t_{A_1}y_{{A_1}3}+y_{{A_1}1}y_{{A_1}2})\sigma_2\\
&+2(-t_{A_1}y_{{A_1}2}+y_{{A_1}1}y_{{A_1}3})\sigma_3,\\
V_{A_1}^{\dagger}\sigma_2V_{A_1}&=(t_{A_1}^2+y_{{A_1}2}^2-y_{{A_1}3}^2-y_{{A_1}1}^2)\sigma_2+2(t_{A_1}y_{{A_1}1}+y_{{A_1}2}y_{{A_1}3})\sigma_3\\
&+2(-t_{A_1}y_{{A_1}3}+y_{{A_1}1}y_{{A_1}2})\sigma_1,\\
V_{A_1}^{\dagger}\sigma_3V_{A_1}&=(t_{A_1}^2+y_{{A_1}3}^2-y_{{A_1}1}^2-y_{{A_1}2}^2)\sigma_3+2(t_{A_1}y_{{A_1}2}+y_{{A_1}1}y_{{A_1}3})\sigma_1\\
&+2(-t_{A_1}y_{{A_1}1}+y_{{A_1}2}y_{{A_1}3})\sigma_2,\\
\end{split}
\end{equation}}}
we get $p_j = \frac{1}{2}(1+(-1)^jsz_3)$, $(j=0,1)$  and
\begin{equation}\label{A2.5}
\begin{split}
\rho_j =& \frac{1}{4(1+(-1)^jsz_3)}V_{A_1}\Pi_jV_{A_1}^\dagger\otimes[\mathbb{I}\otimes \mathbb{I}+(-1)^j\sum_i^3c_iz_i\sigma_i\otimes\sigma_i\\
&+(-1)^jsz_3\mathbb{I}\otimes \mathbb{I}+s\sigma_3\otimes \mathbb{I}+s\mathbb{I}\otimes\sigma_3 ],
\end{split}
\end{equation}
where
\begin{align*}
&z_1=2(-t_{A_1}y_{{A_1}2}+y_{{A_1}1}y_{{A_1}3}),\\
&z_2=2(t_{A_1}y_{{A_1}1}+y_{{A_1}2}y_{{A_1}3}),\\
&z_3=t_{A_1}^2+y_{{A_1}3}^2-y_{{A_1}1}^2-y_{{A_1}2}^2.
\end{align*}

Then $\mathrm{tr}_{A_3}(\rho_j) = \frac{1}{2(1+(-1)^jsz_3)}V_{A_1}\Pi_jV_{A_1}^\dagger\otimes[\mathbb{I}+(-1)^jsz_3\mathbb{I}+s\sigma_3]$, $(j=0,1)$.
The conditional entropy after measurement $\{P^{A_1}_{j}\}$ is given by
\begin{equation}\label{A2.6}
\begin{split}
S_{A_2|\Pi^{A_1}}(\rho) &= -\sum_{j=0}^1p_j(\lambda_j^+\log_2\lambda_j^++\lambda_j^-\log_2\lambda_j^-) \\
&= 1-\frac{1}{4}H(sz_3+s)-\frac{1}{4}H(sz_3-s)+\frac{1}{2}H(sz_3)\\
&=:1-G(z_3).
\end{split}
\end{equation}

In order to evaluate the quantity $S_{A_3|\Pi^{A_1A_2}}(\rho)$, we need to perform measurement on subsystem $A_2$ based on the measurement outcomes 
of $\{P^{A_1}_{j}\}$.
Let $\{P_{jk}^{A_1A_2} = P_{j}^{A_1}\otimes P_{k|j}^{A_2}\}$ denote the bipartite measurement on composite system $A_1A_2$,
where $\{P^{A_2}_{k|j} = V_{A_2}^{(j)}\Pi_k[V_{A_2}^{(j)}]^\dagger,~k=0,1\}$ are the von Neumann measurement on subsystem $A_2$
when the outcome of $\{P_{j}^{A_1}\}$ is $j~ (j = 0, 1)$.
Similarly, 
$V_{A_2}^{(j)} = t_{A_2}^{(j)}\mathbb{I}+\sqrt{-1}\vec{y}_{A_2}^{(j)}\cdot\vec{\sigma}$ is the Bloch expansion of the special unitary matrix satisfying
$t_{A_2}^{(j)}\in \mathbb{R},~ \vec{y}_{A_2}^{(j)} = [y_{{A_2}1}^{(j)},y_{{A_2}2}^{(j)},y_{{A_2}3}^{(j)}]\in \mathbb{R}^3$ and $[y_{{A_2}1}^{(j)}]^2+[y_{{A_2}2}^{(j)}]^2+[y_{{A_2}3}^{(j)}]^2+[t_{A_2}^{(j)}]^2=1$.

If the measurement outcome on system $A_1$ is 0, the measurement $\{P^{A_2}_{k|0}\}$ is performed on state $\rho_0$ in \eqref{A2.5}.
The state after performing the measurement $\{P_{k|0}^{A_2}\}$ reduces to
\begin{equation}\label{A2.7}
\begin{split}
\rho_{0k} &= \frac{1}{2[1+sz_3+(-1)^ksz_3^{(0)}]}V_{A_1}\Pi_0V_{A_1}^{\dagger}\otimes V_{A_2}^{(0)}\Pi_k[V_{A_2}^{(0)}]^\dagger\otimes\Big[ \mathbb{I}\\
&+(-1)^k\sum_i^3c_iz_iz_i^{(0)}\sigma_i+sz_3\mathbb{I}+(-1)^ksz_3^{(0)} \mathbb{I}+s\sigma_3 \Big]
\end{split}
\end{equation}
and $p_{0k} = \frac{1}{4}(1+sz_3+(-1)^ksz_3^{(0)})$ with
$z_1^{(0)}=2(-t_{A_1}^{(0)}y_{{A_1}2}^{(0)}+y_{{A_1}1}^{(0)}y_{{A_1}3}^{(0)}),z_2^{(0)}=2(t_{A_1}^{(0)}y_{{A_1}1}^{(0)}+y_{{A_1}2}^{(0)}y_{{A_1}3}^{(0)}),
z_3^{(0)}=(t_{A_1}^{(0)})^2+(y_{{A_1}3}^{(0)})^2-(y_{{A_1}1}^{(0)})^2-(y_{{A_1}2}^{(0)})^2$.

If the measurement outcome on system $A_1$ is 1, the measurement $\{P^{A_2}_{k|1}\}$ is performed on state $\rho_1$ in Eq.(\ref{A2.5}).
The state after performing the measurement $\{P_{k|1}^{A_2}\}$ reduces to
\begin{equation}\label{A2.8}
\begin{split}
\rho_{1k} &= \frac{1}{2(1-sz_3+(-1)^ksz_3^{(1)})}V_{A_1}\Pi_1V_{A_1}^\dagger\otimes V_{A_2}^{(1)}\Pi_k[V_{A_2}^{(1)}]^\dagger\otimes\Big[ \mathbb{I}\\
&-(-1)^k\sum_i^3c_iz_iz_i^{(1)}\sigma_i-sz_3\mathbb{I}+(-1)^ksz_3^{(1)} \mathbb{I}+s\sigma_3 \Big]
\end{split}
\end{equation}
and $p_{1k} = \frac{1}{4}(1-sz_3+(-1)^ksz_3^{(1)})$ with
$z_1^{(1)}=2(-t_{A_1}^{(1)}y_{{A_1}2}^{(1)}+y_{{A_1}1}^{(1)}y_{{A_1}3}^{(1)}),
z_2^{(1)}=2(t_{A_1}^{(1)}y_{{A_1}1}^{(1)}+y_{{A_1}2}^{(1)}y_{{A_1}3}^{(1)}),
z_3^{(1)}=(t_{A_1}^{(1)})^2+(y_{{A_1}3}^{(1)})^2-(y_{{A_1}1}^{(1)})^2-(y_{{A_1}2}^{(1)})^2$.

The state $\rho_{\Pi^{A_1A_2}}$ is given by $\rho_{\Pi^{A_1A_2}}=\sum_{jk}p_{jk}\rho_{jk}$.
Let $$\varphi_0 = (c_1z_1z_1^{(0)})^2+(c_2z_2z_2^{(0)})^2+(c_3z_3z_3^{(0)})^2$$
and
$$\varphi_1 = (c_1z_1z_1^{(1)})^2+(c_2z_2z_2^{(1)})^2+(c_3z_3z_3^{(1)})^2,$$
then the average entropy of subsystem $A_3$ under the measurement $\{P_{jk}^{A_1A_2}\}$ is given by
\begin{align*}\label{A2.9}
S_{A_3|\Pi^{A_1A_2}}(\rho) = & -\sum_{j=0}^1\sum_{k=0}^1p_{jk}(\lambda_{jk}^+\log_2\lambda_{jk}^++\lambda_{jk}^-\log_2\lambda_{jk}^-) \\
=&1-\frac{1}{8}\Big[H_{sz_3+sz_3^{(0)}}\big(\sqrt{s^2+2sc_3z_3z_3^{(0)}+\varphi_0}\big)\\
&+H_{sz_3-sz_3^{(0)}}\big(\sqrt{s^2-2sc_3z_3z_3^{(0)}+\varphi_0}\big)\\
&+H_{-sz_3+sz_3^{(1)}}\big(\sqrt{s^2+2sc_3z_3z_3^{(1)}+\varphi_1}\big)\\
&+H_{-sz_3-sz_3^{(1)}}\big(\sqrt{s^2-2sc_3z_3z_3^{(1)}+\varphi_1}\big)\\
&-2H_{sz_3}\big(sz_3^{(0)}\big)
-2H_{-sz_3}\big(sz_3^{(1)}\big)\Big]\\
=:&1-F(z_3,z_3^{(0)},z_3^{(1)},\varphi_0,\varphi_1) .
\end{align*}
By \eqref{A1.2}, we can obtain the tripartite quantum discord
\begin{equation}
\begin{split}
\mathcal{Q}&_{A_1,A_2,A_3}(\rho)=\min_{\Pi^{A_1 A_{2}}}[S_{A_2|\Pi^{A_1}}(\rho)+S_{A_3|\Pi^{A_1 A_{2}}}(\rho)-S_{A_2 A_3|A_1}(\rho)]\\
&=\sum_i\lambda_i\log_2\lambda_i+3-\frac{1}{2}H(s)-\max{\{G(z_3)+ F(z_3,z_3^{(0)},z_3^{(1)},\varphi_0,\varphi_1)\}}\\
&=:\sum_i\lambda_i\log_2\lambda_i+3-\frac{1}{2}H(s)-\max{Y(z_3,z_3^{(0)},z_3^{(1)},\varphi_0,\varphi_1)}.
\end{split}
\end{equation}

To solve the maximum of the function $Y$, we first notice that $\vec{z}=(z_1, z_2, z_3)$, $\vec{z}^{(0)}=(z_1^{(0)},z_2^{(0)},z_3^{(0)})$ and $\vec{z}^{(1)}=(z_1^{(1)},z_2^{(1)},z_3^{(1)})$ are independent of each other, so the maximum value of $F$ is obtained when the condition that the function $G$ takes the maximum value. At this point, the variable $z_3$ are certain constants in the function $F$. Namely, the function $F$ is actually only related to $\vec{z}^{(0)}, \vec{z}^{(1)}$. Meanwhile, one note that $F$ is an even function of $z^{(0)}_3$ and $z^{(1)}_3$, respectively, then it is enough to consider $z^{(0)}_3\in [0,1]$ and $z^{(1)}_3\in [0,1]$. Furthermore, we can simplify the optimization to that of a two-variable function.
\begin{theorem}\label{AThm1}
Let $c=\max\{|c_1|,|c_2|\}$, for the family of three-qubit state in \eqref{A2.2}, the quantum discord is explicitly computed as follows.

Case 1 : if $c_3\leq 0$ and $c_3^2\geq c^2$ or $s^2c_3\leq (1-|s|)(c_3^2-c^2)$, then
\begin{equation}
\begin{split}
\mathcal{Q}_{A_1,A_2,A_3}(\rho)=&\sum_i\lambda_i\log_2\lambda_i+3-\frac{1}{8}[H_{2s}(|s+c_3|)\\
&+H_{-2s}(|s-c_3|)+H(|s+c_3|)+H(|s-c_3|)].
\end{split}
\end{equation}

In particular, if $|c_1|=|c_2|=|c_3|=c\in (-1,0]$,
then
\begin{equation}
\begin{split}
\mathcal{Q}_{A_1,A_2,A_3}(\rho)=&\sum_i\lambda_i\log_2\lambda_i+3-\frac{1}{8}[H_{2s}(|s+c|)\\
&+H_{-2s}(|s-c|)+H(|s+c|)+H(|s-c|)],
\end{split}
\end{equation}

Case 2 : if $s=0$, set $C=\max{\{|c_1|,|c_2|,|c_3|\}}$, then we have
\begin{equation}\label{Athm1case2}
\begin{split}
\mathcal{Q}_{A_1,A_2,A_3}(\rho)=\frac{1}{2}H(\sqrt{c_1^2+c_2^2+c_3^2})-\frac{1}{2}H(C),
\end{split}
\end{equation}

\end{theorem}

\begin{proof}
Case 1 :
First, one notices that $Y(z_3,z_3^{(0)},z_3^{(1)},\varphi_0,\varphi_1)$ is a strictly increasing function of $\varphi_0$ and $\varphi_1$:
{
\begin{equation}
\begin{split}
\frac{\partial Y}{\partial \varphi_0}  =& \frac{1}{16\sqrt{s^2+2sc_3z_3z_3^{(0)}+\varphi_0}}\log_2\frac{1+sz_3+sz_3^{(0)}+\sqrt{s^2+2sc_3z_3z_3^{(0)}
+\varphi_0}}{1+sz_3+sz_3^{(0)}-\sqrt{s^2+2sc_3z_3z_3^{(0)}+\varphi_0}}\\
&+\frac{1}{16\sqrt{s^2-2sc_3z_3z_3^{(0)}+\varphi_0}}\log_2\frac{1+sz_3-sz_3^{(0)}+\sqrt{s^2-2sc_3z_3z_3^{(0)}
+\varphi_0}}{1+sz_3-sz_3^{(0)}-\sqrt{s^2-2sc_3z_3z_3^{(0)}+\varphi_0}}\\
\geq & 0
\end{split}
\end{equation}}
and
{
\begin{equation}
\begin{split}
\frac{\partial Y}{\partial \varphi_1}  &= \frac{1}{16\sqrt{s^2+2sc_3z_3z_3^{(1)}+\varphi_1}}\log_2\frac{1-sz_3+sz_3^{(1)}+\sqrt{s^2+2sc_3z_3z_3^{(1)}
+\varphi_1}}{1-sz_3+sz_3^{(1)}-\sqrt{s^2+2sc_3z_3z_3^{(1)}+\varphi_1}}\\
&+\frac{1}{16\sqrt{s^2-2sc_3z_3z_3^{(1)}+\varphi_1}}\log_2\frac{1-sz_3-sz_3^{(1)}+\sqrt{s^2-2sc_3z_3z_3^{(1)}
+\varphi_1}}{1-sz_3-sz_3^{(1)}-\sqrt{s^2-2sc_3z_3z_3^{(1)}+\varphi_1}}\\
&\geq 0.
\end{split}
\end{equation}}
Then $\max {Y}$ takes place at the largest value of $\varphi_0$ and $\varphi_1$ for some $z_3$, $z_3^{(0)}$, $z_3^{(1)}$ $\in [0,1]$.

Set $c = \max{\{|c_1|,|c_2|\}}$, as $(z_1)^2+(z_2)^2+(z_3)^2=1$, $(z_1^{(0)})^2+(z_2^{(0)})^2+(z_3^{(0)})^2=1$ and $(z_1^{(1)})^2+(z_2^{(1)})^2+(z_3^{(1)})^2=1$,
we obtain
\begin{equation}
\begin{split}
\varphi_0 & = (c_1z_1z_1^{(0)})^2+(c_2z_2z_2^{(0)})^2+(c_3z_3z_3^{(0)})^2\\
& \leq (cz_1z_1^{(0)})^2+(cz_2z_2^{(0)})^2+(c_3z_3z_3^{(0)})^2\\
& \leq c^2[1-(z_3z_3^{(0)})^2]+(c_3z_3z_3^{(0)})^2=  :\psi_0. \\
\end{split}
\end{equation}
The last inequality follows from $(z_1z_1^{(0)})^2+(z_2z_2^{(0)})^2+(z_3z_3^{(0)})^2\leq 1$.
For each fixed $z_3,z_3^{(0)}$, the maximum value $\psi_0$ can be achieved by appropriate $z_1$, $z_2$, $z_1^{(0)}$, $z_2^{(0)}$. For example, if $|c_1|\geq |c_2|$, take $z_2=0$ and $z_2^{(0)}=0$, then $\varphi_0 = (cz_1z_1^{(0)})^2+(c_3z_3z_3^{(0)})^2 = \psi_0$.
Similarly,
$$\varphi_1 \leq c^2[1-(z_3z_3^{(1)})^2]+(c_3z_3z_3^{(1)})^2 =: \psi_1.$$
Then $\max Y(z_3,z_3^{(0)},z_3^{(1)},\varphi_0,\varphi_1) = \max Y(z_3,z_3^{(0)},z_3^{(1)},\psi_0,\psi_1)$, where $\psi_j~(j=0,1)$ are function with respect to $z_3$ and $z_3^{(j)}$.

It is not difficult to find that the function $G$ attains the maximum value at $z_3 = 1$ and $\max G = \frac{1}{4}H(2s)-\frac{1}{2}H(s)$.
Since $\vec{z}$, $\vec{z}^{(0)}$ and $\vec{z}^{(1)}$ are independent of each other, the maximum
value of the function $F$ is obtained when the condition that the function $G$ takes the maximum value. The variable $z_3$ are certain constants in the function $F$ at this moment. Namely, the function $F$ is actually only related to $\vec{z}^{(0)}, \vec{z}^{(1)}$. Therefore, $\max{Y(z_3,z_3^{(0)},z_3^{(1)},\varphi_0,\varphi_1)} = \frac{1}{4}H(2s)-\frac{1}{2}H(s)+\max {F(z_3^{(0)},z_3^{(1)})}$, where $F$ is explicitly given by
\begin{equation}
\begin{split}
F&=\frac{1}{8}\Big[H_{s+sz_3^{(0)}}\big(\sqrt{s^2+2sc_3z_3^{(0)}+\psi_0}\big)
+H_{s-sz_3^{(0)}}\big(\sqrt{s^2-2sc_3z_3^{(0)}+\psi_0}\big)\\
&+H_{-s+sz_3^{(1)}}\big(\sqrt{s^2+2sc_3z_3^{(1)}+\psi_1}\big)
+H_{-s-sz_3^{(1)}}\big(\sqrt{s^2-2sc_3z_3^{(1)}+\psi_1}\big)\\
&-2H_s\big(sz_3^{(0)}\big)-2H_{-s}\big(sz_3^{(1)}\big)\Big].
\end{split}
\end{equation}

Next, we compute the partial derivatives of $F$:
{
\begin{equation}\label{AFDZ20}
\begin{split}
\frac{\partial F}{\partial z_3^{(0)}}&=\frac{1}{8}\Bigg[\frac{sc_3-c^2z_3^{(0)}+c_3^2z_3^{(0)}}{K_+}\log_2\frac{1+s+sz_3^{(0)}+K_+}
{1+s+sz_3^{(0)}-K_+}\\
&+\frac{-sc_3-c^2z_3^{(0)}+c_3^2z_3^{(0)}}{K_-}\log_2\frac{1+s-sz_3^{(0)}+K_-}{1+s-sz_3^{(0)}-K_-}\\
&+s\log_2\frac{(1+s+sz_3^{(0)}-K_+)(1+s+sz_3^{(0)}+K_+)(1+s-sz_3^{(0)})^2}{(1+s-sz_3^{(0)}-K_-)(1+s-sz_3^{(0)}
+K_-)(1+s+sz_3^{(0)})^2}\Bigg]\\
&=\frac{1}{8}\Bigg[\frac{sc_3-c^2z_3^{(0)}+c_3^2z_3^{(0)}}{1+s+sz_3^{(0)}}\frac{1}{B_+}\log_2\frac{1+B_+}{1-B_+}\\
&+\frac{-sc_3-c^2z_3^{(0)}+c_3^2z_3^{(0)}}{1+s-sz_3^{(0)}}\frac{1}{B_-}\log_2\frac{1+B_-}{1-B_-}
+s\log_2\frac{1-B_+^2}{1-B_-^2}\Bigg],
\end{split}
\end{equation}}
{
\begin{equation}\label{AFDZ21}
\begin{split}
\frac{\partial F}{\partial z_3^{(1)}}&=\frac{1}{8}[\frac{sc_3-c^2z_3^{(1)}+c_3^2z_3^{(1)}}{L_+}\log_2\frac{1-s+sz_3^{(1)}+L_+}
{1-s+sz_3^{(1)}-L_+}\\
&+\frac{-sc_3-c^2z_3^{(1)}+c_3^2z_3^{(1)}}{L_-}\log_2\frac{1-s-sz_3^{(1)}+L_-}{1-s-sz_3^{(1)}-L_-}\\
&+s\log_2\frac{(1-s+sz_3^{(1)}-L_+)(1-s+sz_3^{(1)}+L_+)(1-s-sz_3^{(1)})^2}{(1-s-sz_3^{(1)}-L_-)(1-s-sz_3^{(1)}
+L_-)(1-s+sz_3^{(1)})^2}]\\
&=\frac{1}{8}[\frac{sc_3-c^2z_3^{(1)}+c_3^2z_3^{(1)}}{1-s+sz_3^{(1)}}\frac{1}{D_+}\log_2\frac{1+D_+}{1-D_+}\\
&+\frac{-sc_3-c^2z_3^{(1)}+c_3^2z_3^{(1)}}{1-s-sz_3^{(1)}}\frac{1}{D_-}\log_2\frac{1+D_-}{1-D_-}+s\log_2\frac{1-D_+^2}{1-D_-^2}],
\end{split}
\end{equation}}
where $K_\pm = \sqrt{(s\pm c_3z_3^{(0)})^2+c^2(1-(z_3^{(0)})^2)}$, $B_\pm = \frac{K_\pm}{1+s\pm sz_3^{(0)}}$,\\ $L_\pm = \sqrt{(s\pm c_3z_3^{(1)})^2+c^2(1-(z_3^{(1)})^2)}$ and $D_\pm = \frac{L_\pm}{1-s\pm sz_3^{(1)}}\in [0,1]$.
Since
{\scriptsize
\begin{equation}
\begin{split}
B_+^2-B_-^2 =\frac{4sz_3^{(0)}\{c_3[(1+s)^2+s^2z_3^{(0)}]-(1+s)[s^2+c_3^2(z_3^{(0)})^2+c^2(1-(z_3^{(0)})^2)]\}}
{(1+s-sz_3^{(0)})^2(1+s+sz_3^{(0)})^2},
\end{split}
\end{equation}} we notice that the last term $s\log_2\frac{1-B_+^2}{1-B_-^2}\geq 0$ of \eqref{AFDZ20} iff $s(B_+^2-B_-^2)\leq 0$, which holds on $c_3\leq 0$.
Since the function $g(x) = \frac{1}{x}\log_2\frac{1+x}{1-x}$ is strictly increasing on $(0,1)$, we have that

(i) If $s\geq 0$, $c_3\leq 0$ and $c_3^2\geq c^2$, then $B_+\leq B_-$. It implies that $\frac{1}{B_+}\log_2\frac{1+B_+}{1-B_+}\leq \frac{1}{B_-}\log_2\frac{1+B_-}{1-B_-}$ and
{\scriptsize
\begin{equation}
\begin{split}\label{AscaleI}
\frac{\partial F}{\partial z_3^{(0)}}
&\geq\frac{1}{8}[\frac{sc_3-c^2z_3^{(0)}+c_3^2z_3^{(0)}}{1+s+sz_3^{(0)}}\frac{1}{B_+}\log_2\frac{1+B_+}{1-B_+}
+\frac{-sc_3-c^2z_3^{(0)}+c_3^2z_3^{(0)}}{1+s-sz_3^{(0)}}\frac{1}{B_-}\log_2\frac{1+B_-}{1-B_-}]\\
&\geq\frac{1}{8}[\frac{sc_3-c^2z_3^{(0)}+c_3^2z_3^{(0)}}{1+s+sz_3^{(0)}}\frac{1}{B_+}\log_2\frac{1+B_+}{1-B_+}
+\frac{-sc_3-c^2z_3^{(0)}+c_3^2z_3^{(0)}}{1+s-sz_3^{(0)}}\frac{1}{B_+}\log_2\frac{1+B_+}{1-B_+}]\\
&\geq\frac{1}{8}[\frac{sc_3-c^2z_3^{(0)}+c_3^2z_3^{(0)}}{1+s+sz_3^{(0)}}\frac{1}{B_+}\log_2\frac{1+B_+}{1-B_+}
+\frac{-sc_3-c^2z_3^{(0)}+c_3^2z_3^{(0)}}{1+s+sz_3^{(0)}}\frac{1}{B_+}\log_2\frac{1+B_+}{1-B_+}]\\
&=\frac{1}{4}[\frac{(c_3^2-c^2)z_3^{(0)}}{1+s+sz_3^{(0)}}\frac{1}{B_+}\log_2\frac{1+B_+}{1-B_+}]\\
&\geq0.
\end{split}
\end{equation}}

(ii) If $s\leq 0$, $c_3\leq 0$ and $c_3^2\geq c^2$, then $B_+\geq B_-$ and we have
{\scriptsize
\begin{equation}\label{AscaleII}
\begin{split}
\frac{\partial F}{\partial z_3^{(0)}}&\geq\frac{1}{8}[\frac{sc_3-c^2z_3^{(0)}+c_3^2z_3^{(0)}}{1+s+sz_3^{(0)}}\frac{1}{B_+}\log_2\frac{1+
B_+}{1-B_+}+\frac{-sc_3-c^2z_3^{(0)}+c_3^2z_3^{(0)}}{1+s-sz_3^{(0)}}\frac{1}{B_-}\log_2\frac{1+B_-}{1-B_-}]\\
&\geq\frac{1}{8}[\frac{sc_3-c^2z_3^{(0)}+c_3^2z_3^{(0)}}{1+s+sz_3^{(0)}}\frac{1}{B_-}\log_2\frac{1+B_-}{1-B_-}
+\frac{-sc_3-c^2z_3^{(0)}+c_3^2z_3^{(0)}}{1+s-sz_3^{(0)}}\frac{1}{B_-}\log_2\frac{1+B_-}{1-B_-}]\\
&\geq\frac{1}{8}[\frac{sc_3-c^2z_3^{(0)}+c_3^2z_3^{(0)}}{1+s-sz_3^{(0)}}\frac{1}{B_-}\log_2\frac{1+B_-}{1-B_-}
+\frac{-sc_3-c^2z_3^{(0)}+c_3^2z_3^{(0)}}{1+s-sz_3^{(0)}}\frac{1}{B_+}\log_2\frac{1+B_-}{1-B_-}]\\
&=\frac{1}{4}[\frac{(c_3^2-c^2)z_3^{(0)}}{1+s-sz_3^{(0)}}\frac{1}{B_-}\log_2\frac{1+B_-}{1-B_-}]\\
&\geq0.
\end{split}
\end{equation}}
Analogous discussion can be done for \eqref{AFDZ21} to get the same results.
Then the maximum of $F(z_3^{(0)},z_3^{(1)})$ is $F(1,1)$ and the maximum of $Y$ is given by
\begin{equation}
\begin{split}
\max{Y}&=\frac{1}{8}[H_{2s}(|s+c_3|)+H_{-2s}(|s-c_3|)+H(|s+c_3|)\\
&+H(|s-c_3|)]-\frac{1}{2}H(s).
\end{split}
\end{equation}

In particular, if $c_1 =c_2 = c_3 =c\in(-1,0]$, by \eqref{AscaleI} and \eqref{AscaleII}, we can get
\begin{equation}
\begin{split}
\max{Y}&=Y(1,1)=\frac{1}{2}H(s)+\frac{1}{8}\big[H_{2s}(|s+c|)+H(|s-c|)\\
&+H(|s+c|)+H_{-2s}(|s-c|)\big].
\end{split}
\end{equation}

Case 2 :
if $s=0$, then
\begin{equation}
\begin{split}
&Y=\frac{1}{4}[H(\sqrt{\varphi_0})+H(\sqrt{\varphi_1})].
\end{split}
\end{equation}

Set $C=\max{\{|c_1|, |c_2|, |c_3|\}}$, we have
\begin{align*}\varphi_0 &= (c_1z_1z_1^{(0)})^2+(c_2z_2z_2^{(0)})^2+(c_3z_3z_3^{(0)})^2\leq C^2\\
\varphi_1 &= (c_1z_1z_1^{(1)})^2+(c_2z_2z_2^{(1)})^2+(c_3z_3z_3^{(1)})^2\leq C^2.
\end{align*}
Note that $H_y(x)$ is an even function of $x$, and is a strictly increasing function on $[0, 1]$,
thus we have
\begin{equation}
\begin{split}
\max{Y}=\frac{1}{2}H(C).
\end{split}
\end{equation}
\end{proof}

Next, we will consider the four-qubit case,
\begin{equation}\label{Afourdef}
\begin{split}
\rho=&\frac{1}{16}(\mathbb{I}+\sum_{i=1}^3c_i\sigma_i\otimes\sigma_i\otimes\sigma_i\otimes\sigma_i+s\sigma_3\otimes \mathbb{I}\otimes \mathbb{I}\otimes \mathbb{I}\\
&+s \mathbb{I}\otimes\sigma_3\otimes \mathbb{I}\otimes \mathbb{I}+ s\mathbb{I}\otimes \mathbb{I}\otimes \mathbb{I}\otimes \sigma_3),
\end{split}
\end{equation}
which is on subsystems $A_1$, $A_2$, $A_3$ and $A_4$.

It can be directly verified that $\rho_{A_1} = \mathrm{tr}_{A_2A_3A_4}=\frac{1}{2}(\mathbb{I}+s\sigma_3)$. Therefore the entropy of the subsystem $A_1$ is $S(\rho_{A_1}) = 1-\frac{1}{2}H(s)$ and we have
$S(\rho) = -\sum_i\lambda_i\log_2\lambda_i$.
Thus,
\begin{equation}
\begin{split}
S_{A_2A_3A_4|A_1}(\rho) = S(\rho)-S(\rho_{A_1}) = -\sum_i\lambda_i\log_2\lambda_i-1+\frac{1}{2}H(s).
\end{split}
\end{equation}

The von Neumann measurement on subsystem $A_1$ can be written as $\{P^{A_1}_{j} = V_{A_1}\Pi_jV_{A_1}^\dagger$, $j=0,1\}$, where $V_{A_1} = t_{A_1}\mathbb{I}+\sqrt{-1}\vec{y}_{A_1}\cdot\vec{\sigma}$
is the unitary operator with
$t_{A_1}\in \mathbb{R}, \vec{y}_{A_1} = (y_{{A_1}1},y_{{A_1}2},y_{{A_1}3})\in \mathbb{R}^3$ and $y_{{A_1}1}^2+y_{{A_1}2}^2+y_{{A_1}3}^2+t_{A_1}^2=1$.
After measuring with $\{P^{A_1}_{j}\}$, the state $\rho$ will become $\rho_{\Pi^{A_1}} = p_0\rho_0+p_1\rho_1$, where
$p_0 = \frac{1}{2}(1+sz_3)$, $p_0 = \frac{1}{2}(1-sz_3)$ and $(j=0,1)$
\begin{equation}
\begin{split}
\rho_j &= \frac{1}{8(1+(-1)^jsz_3)}V_{A_1}\Pi_jV_{A_1}^\dagger\otimes[\mathbb{I}\otimes \mathbb{I}\otimes \mathbb{I}+(-1)^j\sum_i^3c_iz_i\sigma_i\otimes\sigma_i\otimes\sigma_i\\
&+(-1)^jsz_3\mathbb{I}\otimes \mathbb{I}\otimes \mathbb{I}+s\sigma_3\otimes \mathbb{I}\otimes \mathbb{I}+s\mathbb{I}\otimes\sigma_3\otimes\mathbb{I}+s\mathbb{I}\otimes\mathbb{I}\otimes\sigma_3 ]
\end{split}
\end{equation}
with
$z_1=2(-t_{A_1}y_{{A_1}2}+y_{{A_1}1}y_{{A_1}3}),z_2=2(t_{A_1}y_{{A_1}1}+y_{{A_1}2}y_{{A_1}3}),z_3=t_{A_1}^2+y_{{A_1}3}^2-y_{{A_1}1}^2-y_{{A_1}2}^2$.
Therefore, we have $\mathrm{tr}_{A_3A_4}(\rho_0) = \frac{1}{2(1+sz_3)}V_{A_1}\Pi_0V_{A_1}^\dagger\otimes(\mathbb{I}+sz_3\mathbb{I}+s\sigma_3)$ and $\mathrm{tr}_{A_3A_4}(\rho_1) = \frac{1}{2(1-sz_3)}V_{A_1}\Pi_1V_{A_1}^\dagger\otimes(\mathbb{I}-sz_3\mathbb{I}+s\sigma_3)$.
The average entropy entropy of subsystem $A_2$ after measurement on subsystem $A_1$ is given by
\begin{equation}
\begin{split}\label{A4Gz3}
S_{A_2|\Pi^{A_1}}(\rho) =& -\sum_{j=0}^1p_j(\lambda_j^+\log_2\lambda_j^++\lambda_j^-\log_2\lambda_j^-) \\
=& 1-\frac{1}{4}H(sz_3+s)-\frac{1}{4}H(sz_3-s)+\frac{1}{2}H(sz_3)\\
:=&1-G(z_3).
\end{split}
\end{equation}

To evaluate $S_{A_3|\Pi^{A_1A_2}}(\rho)$, one needs to measure the subsystem $A_2$ under the conditional measurement outcome of $A_1$.
We have
\begin{equation}
\begin{split}
\rho_{00} &= \frac{1}{4[1+sz_3+sz_3^{(0)}]}V_{A_1}\Pi_0V_{A_1}^\dagger\otimes
V_{A_2}^{(0)}\Pi_0[V_{A_2}^{(0)}]^\dagger\otimes[\mathbb{I}\otimes\mathbb{I}\\
&+\sum_i^3c_iz_iz_i^{(0)}\sigma_i\otimes\sigma_i+sz_3\mathbb{I}\otimes\mathbb{I}+sz_3^{(0)} \mathbb{I}\otimes\mathbb{I}+s\sigma_3\otimes\mathbb{I}+s\mathbb{I}\otimes\sigma_3 ],
\end{split}
\end{equation}
\begin{equation}
\begin{split}
\rho_{01} &= \frac{1}{4[1+sz_3-sz_3^{(0)}]}V_{A_1}\Pi_1V_{A_1}^\dagger\otimes V_{A_2}^{(0)}\Pi_1[V_{A_2}^{(0)}]^\dagger\otimes[\mathbb{I}\otimes\mathbb{I}\\
&-\sum_i^3c_iz_iz_i^{(0)}\sigma_i\otimes\sigma_i+sz_3\mathbb{I}\otimes\mathbb{I}-sz_3^{(0)} \mathbb{I}\otimes\mathbb{I}+s\sigma_3\otimes\mathbb{I}+s\mathbb{I}\otimes\sigma_3 ],
\end{split}
\end{equation}
\begin{equation}
\begin{split}
\rho_{10} &= \frac{1}{4[1-sz_3+sz_3^{(1)}]}V_{A_1}\Pi_1V_{A_1}^\dagger\otimes {V_{A_2}^{(1)}\Pi_0[V_{A_2}^{(1)}]^\dagger}\otimes[\mathbb{I}\otimes\mathbb{I}\\
&-\sum_i^3c_iz_iz_i^{(1)}\sigma_i\otimes\sigma_i-sz_3\mathbb{I}\otimes\mathbb{I}+sz_3^{(1)} \mathbb{I}\otimes\mathbb{I}+s\sigma_3\otimes\mathbb{I}+s\mathbb{I}\otimes\sigma_3 ],
\end{split}
\end{equation}
\begin{equation}
\begin{split}
\rho_{11} &= \frac{1}{4[1-sz_3-sz_3^{(1)}]}V_{A_1}\Pi_1V_{A_1}^\dagger\otimes {V_{A_2}^{(1)}\Pi_1[V_{A_2}^{(1)}]^\dagger}\otimes[\mathbb{I}\otimes\mathbb{I}\\
&+\sum_i^3c_iz_iz_i^{(1)}\sigma_i\otimes\sigma_i-sz_3\mathbb{I}\otimes\mathbb{I}-sz_3^{(1)} \mathbb{I}\otimes\mathbb{I}+s\sigma_3\otimes\mathbb{I}+s\mathbb{I}\otimes\sigma_3 ]
\end{split}
\end{equation}
and $p_{jk} = \frac{1}{4}(1+(-1)^{j}sz_3+(-1)^ksz_3^{(j)})$
with
\begin{align*}
z_1^{(0)}&=2[-t_{A_1}^{(0)}y_{{A_1}2}^{(0)}+y_{{A_1}1}^{(0)}y_{{A_1}3}^{(0)}],~~
z_2^{(0)}=2[t_{A_1}^{(0)}y_{{A_1}1}^{(0)}+y_{{A_1}2}^{(0)}y_{{A_1}3}^{(0)}],\\
z_3^{(0)}&=[t_{A_1}^{(0)}]^2+[y_{{A_1}3}^{(0)}]^2-[y_{{A_1}1}^{(0)}]^2-[y_{{A_1}2}^{(0)}]^2,\\
z_1^{(1)}&=2[-t_{A_1}^{(1)}y_{{A_1}2}^{(1)}+y_{{A_1}1}^{(1)}y_{{A_1}3}^{(1)}],~~
z_2^{(1)}=2[t_{A_1}^{(1)}y_{{A_1}1}^{(1)}+y_{{A_1}2}^{(1)}y_{{A_1}3}^{(1)}],\\
z_3^{(1)}&=[t_{A_1}^{(1)}]^2+[y_{{A_1}3}^{(1)}]^2-[y_{{A_1}1}^{(1)}]^2-[y_{{A_1}2}^{(1)}]^2,
\end{align*}
where $j$ in $\{V_{A_2}^{(j)},j=0,1\}$ is the outcome of the measurement of the subsystem $A_1$.
$\{V_{A_2}^{(j)} = t_{A_2}^{(j)}\mathbb{I}+\sqrt{-1}\vec{y}_{A_2}^{(j)}\cdot\vec{\sigma}, j=0,1\}$ is the Bloch expansion of unitary operators with $t_{A_2}^{(j)}\in \mathbb{R},~ \vec{y}_{A_2}^{(j)} = [y_{{A_2}1}^{(j)},y_{{A_2}2}^{(j)},y_{{A_2}3}^{(j)}]\in \mathbb{R}^3$ and $[y_{{A_2}1}^{(j)}]^2+[y_{{A_2}2}^{(j)}]^2+[y_{{A_2}3}^{(j)}]^2+[t_{A_2}^{(j)}]^2=1$.

Thus, we have
\begin{align*}
\mathrm{tr}_{A_4}(\rho_{jk})&=\frac{1}{2\big[1+(-1)^jsz_3+(-1)^ksz_3^{(j)}\big]}V_{A_1}\Pi_jV_{A_1}^\dagger\otimes V_{A_2}^{(j)}\Pi_k[V_{A_2}^{(j)}]^\dagger\Big[\mathbb{I}\\
&+(-1)^jsz_3\mathbb{I}+(-1)^ksz_3^{(j)}\mathbb{I}+s\sigma_3\Big]
\end{align*}
and the state $\rho_{\Pi^{A_1A_2}}$ is given by $\rho_{\Pi^{A_1A_2}} = \sum_{jk}p_{jk}\rho_{jk}$.
The average entropy of the subsystem $A_3$ after the measurement $P^{A_1A_2}_{jk}$ is given by
\begin{equation}
\begin{split}\label{A4Fz}
S_{A_3|\Pi^{A_1A_2}}(\rho) &= -\sum_{j=0}^1\sum_{k=0}^1p_{jk}(\lambda_{jk}^+\log_2\lambda_{jk}^++\lambda_{jk}^-\log_2\lambda_{jk}^-) \\
&=1-\frac{1}{8}\Big[H_{sz_3+sz_3^{(0)}}(s)+H_{sz_3-sz_3^{(0)}}(s)+H_{-sz_3+sz_3^{(1)}}(s)\\
&\ \ +H_{-sz_3-sz_3^{(1)}}(s)-2H_{sz_3}(sz_3^{(0)})-2H_{-sz_3}(sz_3^{(1)})\Big]\\
&=:1-F(z_3,z_3^{(0)},z_3^{(1)}).
\end{split}
\end{equation}

To evaluate $S_{A_4|\Pi^{A_1A_2A_3}}(\rho)$, we need to measure the subsystem $A_3$ under the conditions of the outcomes of measuring $A_1$ and $A_2$. We can get $(j,k,l=0,1)$
\begin{equation}
\begin{split}
\rho_{jkl} =& \frac{1}{2\big[1+(-1)^jsz_3+(-1)^ksz_3^{(j)}+(-1)^lsz_3^{(jk)}\big]}V_{A_1}\Pi_jV_{A_1}^\dagger\otimes V_{A_2}^{(j)}\Pi_k(V_{A_2}^{(j)})^\dagger\\
&\otimes V_{A_3}^{(jk)}\Pi_l(V_{A_3}^{(jk)})^\dagger\otimes\Big[\mathbb{I}+(-1)^{j+k+l}\sum_i^3c_iz_iz_i^{(j)}z_i^{(jk)}
\sigma_i+(-1)^jsz_3\mathbb{I}\\
&+(-1)^ksz_3^{(j)} \mathbb{I}+(-1)^lsz_3^{(jk)}\mathbb{I}+s\sigma_3 \Big],
\end{split}
\end{equation}
and $p_{jkl} = \frac{1}{8}(1+(-1)^jsz_3+(-1)^ksz_3^{(j)}+(-1)^lsz_3^{(jk)})$,
where the $j, k$ in the unitary $\{V_{A_3}^{(jk)}, j, k = 0,1 \}$ are the outcome of the measurement of $A_1$ and $A_2$, respectively. And $V_{A_3}^{(jk)}$ can be written as a Bloch expansion $V_{A_3}^{(jk)} = t_{A_3}^{(jk)}\mathbb{I}+\sqrt{-1}\vec{y}_{A_3}^{(jk)}\cdot\vec{\sigma}, j,k=0,1\}$ with $t_{A_3}^{(jk)}\in \mathbb{R},~ \vec{y}_{A_3}^{(jk)} = [y_{{A_3}1}^{(jk)},y_{{A_3}2}^{(jk)},y_{{A_3}3}^{(jk)}]\in \mathbb{R}^3$, $[y_{{A_3}1}^{(jk)}]^2+[y_{{A_3}2}^{(jk)}]^2+[y_{{A_3}3}^{(jk)}]^2+[t_{A_3}^{(jk)}]^2=1$ and
\begin{equation}
\begin{split}
z_1^{(jk)}&=2[-t_{A_1}^{(jk)}y_{{A_1}2}^{(jk)}+y_{{A_1}1}^{(jk)}y_{{A_1}3}^{(jk)}],~~
z_2^{(jk)}=2[t_{A_1}^{(jk)}y_{{A_1}1}^{(jk)}+y_{{A_1}2}^{(jk)}y_{{A_1}3}^{(jk)}],\\
z_3^{(jk)}&=[t_{A_1}^{(jk)}]^2+[y_{{A_1}3}^{(jk)}]^2-[y_{{A_1}1}^{(jk)}]^2-[y_{{A_1}2}^{(jk)}]^2.
\end{split}
\end{equation}
Its nonzero eigenvalues are $(j,k,l=0,1)$
\begin{equation}
\begin{split}
\lambda_{jkl}^{\pm} &= \frac{1+(-1)^jsz_3+(-1)^ksz_3^{(j)}+(-1)^lsz_3^{(jk)}\pm \alpha_{jkl}}{2[1+(-1)^jsz_3+(-1)^ksz_3^{(j)}+(-1)^lsz_3^{(jk)}]}
\end{split}
\end{equation}
with
\begin{equation}
\begin{split}
\alpha_{jkl} =& \bigg[(c_1z_1z_1^{(j)}z_1^{(jk)})^2+(c_2z_2z_2^{(j)}z_2^{(jk)})^2
+(c_3z_3z_3^{(j)}z_3^{(jk)})^2\\
&+(-1)^{j+k+l}2c_3z_3z_3^{(j)}z_3^{(jk)}+s^2\bigg]^{\frac{1}{2}}.
\end{split}
\end{equation}
We have
\begin{equation}
\begin{split}\label{A4Tz}
&S_{A_4|\Pi^{A_1A_2A_3}}(\rho)=
 1-\frac{1}{16}\sum_{jklm}\Big\{[1+(-1)^jsz_3+(-1)^ksz_3^{(j)}+(-1)^lsz_3^{(jk)}\\
&+(-1)^m\alpha_{jkl}]\log_2[1+(-1)^jsz_3+(-1)^ksz_3^{(j)}+(-1)^lsz_3^{(jk)}
+(-1)^m\alpha_{jkl}\\
&-2[1+(-1)^jsz_3+(-1)^ksz_3^{(j)}+(-1)^lsz_3^{(jk)}]\log_2[1+(-1)^jsz_3+(-1)^ksz_3^{(j)}\\
&+(-1)^lsz_3^{(jk)}]\Big\}=:1-T,
\end{split}
\end{equation}
where $T$ is a function with respect to $z_3,z_3^{(0)},z_3^{(1)},z_3^{(00)},z_3^{(01)},z_3^{(10)},z_3^{(11)},\alpha_{jkl} \\(j,k,l=0,1)$.

By \eqref{A1.2}, \eqref{A4Fz}, \eqref{A4Gz3} and \eqref{A4Tz}, we can obtain the four-partite quantum discord
\begin{equation}
\begin{split}
&\mathcal{Q}_{A_1,A_2,A_3,A_4}(\rho)\\
&=\min_{\Pi^{A_1 A_2A_3}}[S_{A_2|\Pi^{A_1}}(\rho)+S_{A_3|\Pi^{A_1 A_{2}}}(\rho)+S_{A_4|\Pi^{A_1 A_2A_3}}(\rho)-S_{A_2 A_3|A_1}(\rho)]\\
&=\sum_i\lambda_i\log_2\lambda_i+4-\frac{1}{2}H(s)-\max{\{G+ F+T\}}\\
&=:\sum_i\lambda_i\log_2\lambda_i+4-\frac{1}{2}H(s)-\max{Y}.
\end{split}
\end{equation}

\begin{theorem}
For the family of four-qubit state in \eqref{Afourdef}, the quantum discord is evaluated in two regions.

Case 1 : if $c_3\leq 0$ and $c_3^2\geq c^2$ or $\frac{s^2}{(1-2|s|)}\geq \frac{(c_3^2-c^2)}{c_3}$, then we have
\begin{equation}
\begin{split}
\mathcal{Q}_{A_1,A_2,A_3,A_4}(\rho)&=\sum_i\lambda_i\log_2\lambda_i+4-\frac{1}{16}\big[H_{3s}(|s+c_3|)+H_{-3s}(|s-c_3|)\\
&+3H_{s}(|s-c_3|)+3H_{-s}(|s+c_3|)\big].
\end{split}
\end{equation}

Case 2 : if $s=0$, set $C=\max{\{|c_1|,|c_2|,|c_3|\}}$, then
\textcolor{red}{\begin{equation}\label{Athm2case2}
\begin{split}
\mathcal{Q}_{A_1,A_2,A_3,A_4}(\rho)=&\frac{1}{4}\Big[(1+c_1+c_2+c_3)\log_2(1+c_1+c_2+c_3)\\
&+(1+c_1-c_2-c_3)\log_2(1+c_1-c_2-c_3)\\
&+(1-c_1+c_2-c_3)\log_2(1-c_1+c_2-c_3)\\
&+(1-c_1-c_2+c_3)\log_2(1-c_1-c_2+c_3)
\Big]\\
&-\frac{1}{2}H(C).
\end{split}\end{equation}}
\end{theorem}

\begin{proof}

The proof is similar to that Theorem \ref{AThm1}. First, since
$\vec{z}=(z_1, z_2, z_3)$, $\vec{z}^{(0)}=(z_1^{(0)},z_2^{(0)},z_3^{(0)})$, $\vec{z}^{(1)}=(z_1^{(1)},z_2^{(1)},z_3^{(1)})$ $\vec{z}^{(00)}=(z_1^{(00)},z_2^{(00)},z_3^{(00)})$, $\vec{z}^{(01)}=(z_1^{(01)},z_2^{(01)},z_3^{(01)})$ $\vec{z}^{(10)}=(z_1^{(10)},z_2^{(10)},$ $z_3^{(10)})$ and $\vec{z}^{(11)}=(z_1^{(11)},z_2^{(11)},z_3^{(11)})$ are independent of each other, the maximum value of $T$ is obtained when both functions $G$ and $F$ take the maximum value.
Thus, $\max Y = \frac{1}{8}H(3s)-\frac{1}{8}H(s)+\max{T}$.
Meanwhile, similar to the proof of Theorem \ref{AThm1}, we can reduce this optimization problem to that of a four-variable function as follows.
\begin{equation}
\begin{split}
&T(z_3^{(00)},z_3^{(01)},z_3^{(10)},z_3^{(11)})=\frac{1}{16}\sum_{jklm}\Big[1+(-1)^js+(-1)^ks+(-1)^lsz_3^{(jk)}\\
&+(-1)^m\sqrt{s^2+(-1)^{j+k+l}c_3z_3^{(jk)}+\beta_{jk}}\Big]\log_2\Big[1+(-1)^js+(-1)^ks+(-1)^lsz_3^{(jk)}\\
&+(-1)^m\sqrt{s^2+(-1)^{j+k+l}c_3z_3^{(jk)}+\beta_{jk}}\Big]-\frac{1}{8}\sum_{jkl}\Big[1+(-1)^js+(-1)^ks\\
&+(-1)^lsz_3^{(jk)}]\log_2[1+(-1)^js+(-1)^ks+(-1)^lsz_3^{(jk)}\Big],
\end{split}
\end{equation}
where $\beta_{jk} = c^2\big[1-(z_3^{(jk)})^2\big]+(c_3z_3^{(jk)})^2, ~(j,k=0,1)$.

Taking the derivative of $T$ over $z_3^{(00)}$, $z_3^{(01)}$, $z_3^{(10)}$ and $z_3^{(11)}$, respectively. We can get that $(j,k=0,1)$
{\scriptsize
\begin{equation}\label{ATDZjk3}
\begin{split}
\frac{\partial T}{\partial z_3^{(jk)}}=&\frac{1}{16}\Bigg[s\log_2\frac{\left(U_{jk}+sz_3^{(jk)}-Q_{jk0}\right)\left(U_{jk}+sz_3^{(jk)}
+Q_{jk0}\right)\left(U_{jk}-sz_3^{(jk)}\right)^2}{\left(U_{jk}+sz_3^{(jk)}\right)^2\left(U_{jk}-sz_3^{(jk)}
-Q_{jk1}\right)\left(U_{jk}-sz_3^{(jk)}+Q_{jk1}\right)}\\
&+\frac{sc_3-c^2z_3^{(jk)}+c_3^2z_3^{(jk)}}{Q_{jk0}}\log_2\frac{U_{jk}+sz_3^{(jk)}+Q_{jk0}}{U_{jk}+sz_3^{(jk)}-Q_{jk0}}
\\&+\frac{-sc_3-c^2z_3^{(jk)}+c_3^2z_3^{(jk)}}{Q_{jk1}}\log_2\frac{U_{jk}-sz_3^{(jk)}+Q_{jk1}}{U_{jk}-sz_3^{(jk)}-Q_{jk1}}\Bigg]\\
=&\frac{1}{16}\Bigg[\frac{sc_3-c^2z_3^{(jk)}+c_3^2z_3^{(jk)}}{U_{jk}+sz_3^{(jk)}}\frac{1}{E_{jk0}}\log_2\frac{1+E_{jk0}}{1-E_{jk0}}
\\
&+\frac{-sc_3-c^2z_3^{(jk)}+c_3^2z_3^{(jk)}}{U_{jk}-sz_3^{(jk)}}\frac{1}{E_{jk1}}\log_2\frac{1+E_{jk1}}{1-E_{jk1}}
+s\log_2\frac{1-(E_{jk0})^2}{1-(E_{jk1})^2}\Bigg],
\end{split}
\end{equation}}
where $Q_{jkl} = \sqrt{\left(s+(-1)^{j+k+l} c_3z_3^{(jk)}\right)^2+c^2\left(1-(z_3^{(jk)})^2\right)}$, $U_{jk}=1+(-1)^js+(-1)^ks$ and $E_{jkl} = \frac{Q_{jkl}}{\left(1+s+(-1)^l sz_3^{(jk)}\right)}\in [0,1]$, $(j,k,l=0,1)$.
It is not difficult to see that the derivative of $T$ over $z_3^{(jk)}$ is similar to \eqref{AFDZ20}, so a similar discussion can be conducted.

Case 1 : Since
{\small
\begin{equation}
\begin{split}
&(E_{jk0})^2-(E_{jk1})^2 =\\ &\frac{4sz_3^{(jk)}\left\{c_3\left[(U_{jk})^2+s^2z_3^{(jk)}\right]-U_{jk}\left[s^2+c_3^2(z_3^{(jk)})^2
+c^2\left(1-(z_3^{(jk)})^2\right)\right]\right\}}{\left(U_{jk}-sz_3^{(jk)}\right)^2\left(U_{jk}+sz_3^{(jk)}\right)^2}
\end{split}
\end{equation}}
we see that the last term $s\log_2\frac{1-(E_{jk0})^2}{1-(E_{jk1})^2}\geq 0$ of \eqref{ATDZjk3}) iff $s\left[(E_{jk0})^2-(E_{jk1})^2\right]\leq 0$, which holds on $c_3\leq 0$.
Since function $g(x) = \frac{1}{x}\log_2\frac{1+x}{1-x}$ is strictly increasing on $(0,1)$, we have that

(i) If $s\geq 0$, $c_3\leq 0$ and $c_3^2\geq c^2$, then $E_{jk0}\leq E_{jk1}$. It implies that $\frac{1}{E_{jk0}}\log_2\frac{1+E_{jk0}}{1-E_{jk0}}\leq \frac{1}{E_{jk1}}\log_2\frac{1+E_{jk1}}{1-E_{jk1}}$ and
{\scriptsize
\begin{equation}
\begin{split}
\frac{\partial T}{\partial z_3^{(jk)}}&\geq\frac{1}{16}\left[\frac{sc_3-c^2z_3^{(jk)}+c_3^2z_3^{(jk)}}{\left(U_{jk}+sz_3^{(jk)}\right)
E_{jk0}}\log_2\frac{1+E_{jk0}}{1-E_{jk0}}+\frac{-sc_3-c^2z_3^{(jk)}+c_3^2z_3^{(jk)}}{\left(U_{jk}-sz_3^{(jk)}\right)
E_{jk1}}\log_2\frac{1+E_{jk1}}{1-E_{jk1}}\right]\\
&\geq\frac{1}{16}\left[\frac{sc_3-c^2z_3^{(jk)}+c_3^2z_3^{(jk)}}{\left(U_{jk}+sz_3^{(jk)}\right)
E_{jk0}}\log_2\frac{1+E_{jk0}}{1-E_{jk0}}+\frac{-sc_3-c^2z_3^{(jk)}+c_3^2z_3^{(jk)}}{\left(U_{jk}-sz_3^{(jk)}\right)
E_{jk0}}\log_2\frac{1+E_{jk0}}{1-E_{jk0}}\right]\\
&\geq\frac{1}{16}\left[\frac{sc_3-c^2z_3^{(jk)}+c_3^2z_3^{(jk)}}{\left(U_{jk}+sz_3^{(jk)}\right)
E_{jk0}}\log_2\frac{1+E_{jk0}}{1-E_{jk0}}+\frac{-sc_3-c^2z_3^{(jk)}+c_3^2z_3^{(jk)}}{\left(U_{jk}+sz_3^{(jk)}\right)
E_{jk0}}\log_2\frac{1+E_{jk0}}{1-E_{jk0}}\right]\\
&=\frac{1}{8}\left[\frac{(c_3^2-c^2)z_3^{(jk)}}{U_{jk}+sz_3^{(jk)}}\frac{1}{E_{jk0}}\log_2\frac{1+E_{jk0}}{1-E_{jk0}}\right]\\
&\geq0.
\end{split}
\end{equation}}

(ii) If $s\geq 0$, $c_3\leq 0$ and $\frac{s^2}{(1-2|s|)}\geq \frac{(c_3^2-c^2)}{c_3}$, then
{\scriptsize
\begin{equation}
\begin{split}
\frac{\partial T}{\partial z_3^{(jk)}}&\geq\frac{1}{16}\left[\frac{sc_3-c^2z_3^{(jk)}+c_3^2z_3^{(jk)}}{\left(U_{jk}+sz_3^{(jk)}\right)
E_{jk0}}\log_2\frac{1+E_{jk0}}{1-E_{jk0}}+\frac{-sc_3-c^2z_3^{(jk)}+c_3^2z_3^{(jk)}}{\left(U_{jk}-sz_3^{(jk)}\right)
E_{jk1}}\log_2\frac{1+E_{jk1}}{1-E_{jk1}}\right]\\
&\geq\frac{1}{16}\left[\frac{sc_3-c^2z_3^{(jk)}+c_3^2z_3^{(jk)}}{\left(U_{jk}+sz_3^{(jk)}\right)
E_{jk0}}\log_2\frac{1+E_{jk0}}{1-E_{jk0}}+\frac{-sc_3-c^2z_3^{(jk)}+c_3^2z_3^{(jk)}}{\left(U_{jk}-sz_3^{(jk)}\right)
E_{jk0}}\log_2\frac{1+E_{jk0}}{1-E_{jk0}}\right]\\
&=\frac{1}{16}\left[\frac{2z_3^{(jk)}\left(U_{jk}(c_3^2-c^2)-s^2c_3\right)}{\left(U_{jk}+sz_3^{(jk)}\right)
\left(U_{jk}-sz_3^{(jk)}\right)}\frac{1}{E_{jk0}}\log_2\frac{1+E_{jk0}}{1-E_{jk0}}\right]\\
&\geq0.\\
\end{split}
\end{equation}}

(iii) If $s\leq 0$, $c_3\leq 0$ and $c_3^2\geq c^2$, then $E_{jk0}\geq E_{jk1}$ and we have
{\scriptsize
\begin{equation}
\begin{split}
\frac{\partial T}{\partial z_3^{(jk)}}&\geq\frac{1}{16}\left[\frac{sc_3-c^2z_3^{(jk)}+c_3^2z_3^{(jk)}}{\left(U_{jk}+sz_3^{(jk)}\right)
E_{jk0}}\log_2\frac{1+E_{jk0}}{1-E_{jk0}}+\frac{-sc_3-c^2z_3^{(jk)}+c_3^2z_3^{(jk)}}{\left(U_{jk}-sz_3^{(jk)}\right)
E_{jk1}}\log_2\frac{1+E_{jk1}}{1-E_{jk1}}\right]\\
&\geq\frac{1}{16}\left[\frac{sc_3-c^2z_3^{(jk)}+c_3^2z_3^{(jk)}}{\left(U_{jk}+sz_3^{(jk)}\right)
E_{jk1}}\log_2\frac{1+E_{jk1}}{1-E_{jk1}}+\frac{-sc_3-c^2z_3^{(jk)}+c_3^2z_3^{(jk)}}{\left(U_{jk}-sz_3^{(jk)}\right)
E_{jk1}}\log_2\frac{1+E_{jk1}}{1-E_{jk1}}\right]\\
&\geq\frac{1}{16}\left[\frac{sc_3-c^2z_3^{(jk)}+c_3^2z_3^{(jk)}}{\left(U_{jk}-sz_3^{(jk)}\right)
E_{jk1}}\log_2\frac{1+E_{jk1}}{1-E_{jk1}}+\frac{-sc_3-c^2z_3^{(jk)}+c_3^2z_3^{(jk)}}{\left(U_{jk}-sz_3^{(jk)}\right)
E_{jk1}}\log_2\frac{1+E_{jk1}}{1-E_{jk1}}\right]\\
&=\frac{1}{8}\left[\frac{(c_3^2-c^2)z_3^{(jk)}}{\left(U_{jk}-sz_3^{(jk)}\right)
E_{jk1}}\log_2\frac{1+E_{jk1}}{1-E_{jk1}}\right]\\
&\geq0.
\end{split}
\end{equation}}

(iv)If $s\leq 0$, $c_3\leq 0$ and $\frac{s^2}{(1-2|s|)}\geq \frac{(c_3^2-c^2)}{c_3}$, then
{\scriptsize
\begin{equation}
\begin{split}
\frac{\partial T}{\partial z_3^{(jk)}}&\geq\frac{1}{16}\left[\frac{sc_3-c^2z_3^{(jk)}+c_3^2z_3^{(jk)}}{\left(U_{jk}+sz_3^{(jk)}\right)
E_{jk0}}\log_2\frac{1+E_{jk0}}{1-E_{jk0}}+\frac{-sc_3-c^2z_3^{(jk)}+c_3^2z_3^{(jk)}}{\left(U_{jk}-sz_3^{(jk)}\right)
E_{jk1}}\log_2\frac{1+E_{jk1}}{1-E_{jk1}}\right]\\
&\geq\frac{1}{16}\left[\frac{sc_3-c^2z_3^{(jk)}+c_3^2z_3^{(jk)}}{\left(U_{jk}+sz_3^{(jk)}\right)
E_{jk1}}\log_2\frac{1+E_{jk1}}{1-E_{jk1}}+\frac{-sc_3-c^2z_3^{(jk)}+c_3^2z_3^{(jk)}}{\left(U_{jk}-sz_3^{(jk)}\right)
E_{jk1}}\log_2\frac{1+E_{jk1}}{1-E_{jk1}}\right]\\
&=\frac{1}{8}\left[\frac{z_3^{(jk)}\left(U_{jk}(c_3^2-c^2)-s^2c_3\right)}{\left(U_{jk}+sz_3^{(jk)}\right)
\left(U_{jk}-sz_3^{(jk)}\right)}\frac{1}{E_{jk1}}\log_2\frac{1+E_{jk1}}{1-E_{jk1}}\right]\\
&\geq0.
\end{split}
\end{equation}}

Then the maximum of $T(z_3^{(00)},z_3^{(01)},z_3^{(10)},z_3^{(11)})$ is $T(1,1,1,1)$, and
\begin{equation}
\begin{split}\max{Y} =& \frac{1}{16}[H_{3s}(|s+c_3|)+H_{-3s}(|s-c_3|)+3H_{s}(|s-c_3|)\\
&+3H_{-s}(|s+c_3|)]-\frac{1}{2}H(s)
.\end{split}
\end{equation}

Case 2 :
if $s=0$, then
\begin{equation}
\begin{split}
&Y=\frac{1}{16}\sum_{j=0}^1\sum_{k=0}^1\sum_{l=0}^1H(\alpha_{jkl}).
\end{split}
\end{equation}
Set $C=\max{\{|c_1|, |c_2|, |c_3|\}}$, we have
$$\alpha_{jkl} = \sqrt{(c_1z_1z_1^{(j)}z_1^{(jk)})^2+(c_2z_2z_2^{(j)}z_2^{(jk)})^2+(c_3z_3z_3^{(j)}z_3^{(jk)})^2}\leq C.$$
Note that $H_y(x)$ is an even function of $x$, and is strictly increasing on $[0, 1]$,
thus we have
\begin{equation}
\begin{split}
\max{Y}=\frac{1}{2}H(C).
\end{split}
\end{equation}

\end{proof}

From the results of three-qubit and four-quibt states, we obtain the following conclusion for the general N-qubit
case.
\begin{theorem}\label{AThmN}
For the family of N-qubit states in \eqref{A2.1}, the quantum discord is given as follows.

Case 1: if $c_3\leq 0$ , $c^2\leq c_3^2$ or $\frac{s^2}{1-(N-2)|s|}\geq \frac{c_3^2-c^2}{c_3}$, then
\begin{equation}
\begin{split}
\mathcal{Q}_{A_1,A_2,\ldots,A_N}(\rho)&=\sum_i\lambda_i\log_2\lambda_i+N-\max{W},
\end{split}
\end{equation}
where
$\max{W}$ is computed according to the following cases.\\
(i) if $N=4n$, $n\in\mathbb{N}^+$, then
{\small
\begin{equation}\label{A4n}
\begin{split}
\max{W} &= \frac{1}{2^{4n}}\left[\sum_{k=0}^{2n-1}\tbinom{4n-1}{2k}H_{(4n-4k-1)s}(|s+c_3|)+\sum_{k=0}^{2n-1}\tbinom{4n-1}{2k+1}H_{(4n-4k-3)s}(|s-c_3|)\right];
\end{split}
\end{equation}}
(ii) if $N=4n+1$, $n\in\mathbb{N}^+$, then
{\small
\begin{equation}\label{A4n+1}
\begin{split}
\max{W} &= \frac{1}{2^{4n+1}}\left[\sum_{k=0}^{2n}\tbinom{4n}{2k}H_{(4n-4k)s}(|s+c_3|)+\sum_{k=0}^{2n-1}\tbinom{4n}{2k+1}H_{(4n-4k-2)s}(|s-c_3|)\right];
\end{split}
\end{equation}}
(iii) if $N=4n+2$, $n\in\mathbb{N}^+$, then
{\small
\begin{equation}\label{A4n+2}
\begin{split}
\max{W} &= \frac{1}{2^{4n+2}}\left[\sum_{k=0}^{2n}\tbinom{4n+1}{2k}H_{(4n-4k+1)s}(|s+c_3|)+\sum_{k=0}^{2n}\tbinom{4n+1}{2k+1}H_{(4n-4k-1)s}(|s-c_3|)\right];
\end{split}
\end{equation}}
(iv) if $N=4n+3$, $n\in\mathbb{N}^+$, then
{\small
\begin{equation}\label{A4n+3}
\begin{split}
\max{W} &= \frac{1}{2^{4n+3}}\left[\sum_{k=0}^{2n+1}\tbinom{4n+2}{2k}H_{(4n-4k+2)s}(|s+c_3|)+\sum_{k=0}^{2n}\tbinom{4n+2}{2k+1}H_{(4n-4k)s}(|s-c_3|)\right],
\end{split}
\end{equation}}
where $\tbinom{N-1}{j}$ stands for the binomial number.\\
Case 2: if $s=0$, set $C=\max{\{|c_1|,|c_2|,|c_3|\}}$, then
\begin{equation}
\begin{split}
\mathcal{Q}_{A_1,A_2,\ldots,A_N}(\rho)&=\sum_i\lambda_i\log_2\lambda_i+N-\frac{1}{2}H(C).
\end{split}
\end{equation}
\end{theorem}

\begin{proof}
Based on the previous results of 3-qubits and 4-qubits, we can know that
$S(\rho_{A_1})=1-\frac{1}{2}H(s)$ and $(k=2,\ldots,N-1)$
{\small
\begin{equation}
\begin{split}
&S_{A_k|\Pi^{A_1\ldots A_{k-1}}}(\rho) =1-\frac{1}{2^k}\sum_{u_i=0,1}\bigg[\Big(1+(-1)^{u_1}sz_3+(-1)^{u_2}sz_3^{(u_1)}+(-1)^{u_3}sz_3^{(u_1u_2)}+\cdots\\
&+(-1)^{u_{k-1}}sz_3^{(u_1\cdots \mu_{k-2})}+(-1)^{u_k}s\Big)\log_2\Big(1+(-1)^{u_1}sz_3+(-1)^{u_2}sz_3^{(u_1)}+(-1)^{u_3}sz_3^{(u_1u_2)}\\
&+\cdots+(-1)^{u_{k-1}}sz_3^{(u_1\cdots \mu_{k-2})}+(-1)^{u_k}s\Big)\bigg]+\frac{1}{2^{k-1}}\sum_{u_i=0,1}\bigg[\Big(1+(-1)^{u_1}sz_3+(-1)^{u_2}sz_3^{(u_1)}\\
&+(-1)^{u_3}sz_3^{(u_1u_2)}+\cdots+(-1)^{u_{k-1}}sz_3^{(u_1\cdots u_{k-2})}\Big)\log_2\Big(1+(-1)^{u_1}sz_3+(-1)^{u_2}sz_3^{(u_1)}\\
&+(-1)^{u_3}sz_3^{(u_1u_2)}+\cdots+(-1)^{u_{k-1}}sz_3^{(u_1\cdots u_{k-2})}\Big)\bigg] ,
\end{split}
\end{equation}}
where $u_i$ stands for the outcome of the measurement of subsystem $A_i$ and the $\sum_{u_i=0,1}$ is a summation over all possible combinations of $u_1 = 0, 1$, $u_2 = 0, 1,\ldots ,u_{k} = 0, 1$.
Furthermore, since $z_3,z_3^{(u_1)},z_3^{(u_1u_2)},\ldots,z_3^{(u_1u_2\ldots u_{k-2})}$ are independent of each other, we have
\begin{equation}
\begin{split}\min_{\Pi^{A_1\ldots A_{N-1}}}&\{S_{A_2|\Pi^{A_1}}(\rho)+\cdots+S_{A_{N}|\Pi^{A_1\ldots A_{N-1}}}(\rho)\}\\
&=\min{S_{A_2|\Pi^{A_1}}(\rho)}+\cdots+\min{S_{A_{N}|\Pi^{A_1\ldots A_{N-1}}}(\rho)}.
\end{split}
\end{equation}
Meanwhile,
for $k=2,\ldots,N-1$, the minimum value of $S_{A_k|\Pi^{A_1\ldots A_{k-1}}}(\rho)$ is attained at $z_3^{(u_1)}=1$, $z_3^{(u_1u_{2})}=1$,\ldots,$z_3^{(u_1\cdots u_{k-2})}=1$. So
\begin{equation}
\begin{split}
S(\rho_{A_1})&=1-\frac{1}{2}H(s):=1-S_1;\\
\min{S_{A_2|\Pi^{A_1}}(\rho)}&=1-\frac{1}{4}H(2s)+\frac{1}{2}H(s):=1-S_2+S_1;\\
&\cdots\\
\min{S_{A_{N-1}|\Pi^{A_1\ldots A_{N-2}}}(\rho) } &=1-\frac{1}{2^{N-1}}\sum_{j=0}^{N-1}\tbinom{N-1}{j}H\big((N-1-2j)s\big)\\
&+\frac{1}{2^{N-2}}\sum_{j=0}^{N-2}\tbinom{N-2}{j}H\big((N-2-2j)s\big)\\
:&=1-S_{N-1}+S_{N-2};\\
\min{S_{A_{N}|\Pi^{A_1\ldots A_{N-1}}}(\rho)} &= 1-\max{W}.
\end{split}
\end{equation}
where
\begin{equation}
\begin{split}
W=&\frac{1}{2^N}\sum_{u_i=0,1}\bigg[1+(-1)^{u_{N-1}}sz_3^{(u_1\cdots u_{N-2})}+\sum_{j=1}^{N-2}(-1)^{u_j}s+(-1)^{u_N}\xi_{u_1\cdots u_{N-1}}\bigg]\\&\cdot\log_2\bigg[1+(-1)^{u_{N-1}}sz_3^{(u_1\cdots u_{N-2})}+\sum_{j=1}^{N-2}(-1)^{u_j}s+(-1)^{u_N}\xi_{u_1\cdots u_{N-1}}\bigg]\\&+\frac{1}{2^{N-1}}\sum_{u_i=0,1}\bigg[1+\sum_{j=1}^{N-2}(-1)^{u_j}s+(-1)^{u_{N-1}}sz_3^{(u_1\cdots u_{N-2})}\bigg]\\&\cdot\log_2\bigg[1
+\sum_{j=1}^{N-2}(-1)^{u_j}s+(-1)^{u_{N-1}}sz_3^{(u_1\cdots u_{N-2})}\bigg]
\end{split}
\end{equation}
 with
 {\small$$\xi_{u_1\cdots u_{N-1}} = \sqrt{(c_1z_1^{(u_1\cdots u_{N-2})})^2+(c_2z_2^{(u_1\cdots u_{N-2})})^2+(s+(-1)^{\sum_{j=1}^{N-1}u_j}c_3z_3^{(u_1\cdots u_{N-2})})^2}.$$}
It follows from \eqref{A1.2} that
\begin{equation}
\begin{split}
\mathcal{Q}_{A_1,A_2,\ldots,A_N}(\rho) &= \sum_i\lambda_i\log_2\lambda_i+N-\frac{1}{2^{N-1}}\sum_{j=0}^{N-1}\tbinom{N-1}{j}H\big((N-1-2j)s\big)\\
&-\max{W}.
\end{split}
\end{equation}
Similar to the previous discussion of 3- and 4-qubits, it is not difficult to get the following two results.

Case 1: if $c_3\leq 0$ and $c_3^2\geq c^2$ or $\frac{s^2}{1-(N-2)|s|}\geq \frac{c_3^2-c^2}{c_3}$,
{
\begin{equation}
\begin{split}
\max{W}=&\frac{1}{2^N}\sum_{u_k=0,1}\bigg[\Big(1+\sum_{k=1}^{N-1} (-1)^{u_k}s+(-1)^{u_N}\Big|s+(-1)^{\sum_{k=1}^{N-1}u_k}c_3\Big|\Big)\log_2\Big(1\\
&+\sum_{k=1}^{N-1} (-1)^{u_k}s+(-1)^{u_N}\Big|s+(-1)^{\sum_{k=1}^{N-1}u_k}c_3\Big|\Big)\bigg]-\frac{1}{2^{N-1}}\sum_{u_k=0,1}\Big(1\\
&+\sum_{k=1}^{N-1} (-1)^{u_k}s\Big)\log_2\Big(1+\sum_{k=1}^{N-1} (-1)^{u_k}s\Big),
\end{split}
\end{equation}}
where the $\sum_{u_k=0,1}$ stands for a summation over all possible combinations of $u_1 = 0, 1$, $u_2 = 0, 1\ldots ,u_{N} = 0, 1$.
By further discussing by dividing into four cases $N=4n$, $N=4n+1$,$N=4n+2$, $N=4n+3$, $n\in\mathbb{N}^+$, the results in Theorem \ref{AThmN} are proved.

Case 2: if $s=0$, set $C=\max{\{|c_1|,|c_2|,|c_3|\}}$, then
\begin{equation}
\begin{split}
&\max{W}=\frac{1}{2}[H(C)].
\end{split}
\end{equation}
\end{proof}

\begin{theorem}\label{AThmN2}
For the following family of N-qubit states,
\begin{equation}
\begin{split}
\rho = \frac{1}{2^N}(\mathbb{I}+s_1\sigma_3\otimes \mathbb{I}\cdots\otimes \mathbb{I}+s_2\mathbb{I}\otimes \sigma_3\otimes\mathbb{I}\cdots\otimes \mathbb{I}+\cdots+s_N\mathbb{I}\otimes\cdots\otimes\mathbb{I}\otimes\sigma_3),
\end{split}
\end{equation}
the quantum discord is given by
\begin{equation}
\begin{split}
\mathcal{Q}_{A_1,A_2,\ldots,A_N}(\rho) = \sum_i\lambda_i\log_2\lambda_i+N-\frac{1}{2^{N}}\sum_{u_i=0,1}H_{\sum_{i=1}^{N-1}(-1)^{u_i}s_i}(|s_{N}|).
\end{split}
\end{equation}
\end{theorem}

\begin{proof}
According to the previous steps,
we follow the same strategy to obtain that
\begin{equation}
\begin{split}
S(\rho_{A_1})&=1-\frac{1}{2}H(s_1):=1-S_1;\\
\min{S_{A_2|\Pi^{A_1}}(\rho)}&=1-\frac{1}{4}H_{s_1}(s_2)-\frac{1}{4}H_{-s_1}(s_2)+\frac{1}{2}H(s_1)\\
:&=1-S_2+S_1;\\
&\ \ \cdots\\
\min{S_{A_{N-1}|\Pi^{A_1\ldots A_{N-2}}}(\rho)} &=1-\frac{1}{2^{N-1}}\sum_{u_i=0,1}H_{\sum_{i=1}^{N-2}(-1)^{u_i}s_i}(s_{N-1})\\
&+\frac{1}{2^{N-2}}\sum_{u_i=0,1}H_{\sum_{i=1}^{N-3}(-1)^{u_i}s_i}(s_{N-2})\\
:&=1-S_{N-1}+S_{N-2};\\
\min{S_{A_{N}|\Pi^{A_1\ldots A_{N-1}}}(\rho)} &= 1-\frac{1}{2^{N}}\sum_{u_i=0,1}H_{\sum_{i=1}^{N-1}(-1)^{u_i}s_i}(|s_{N}|)\\
&+\frac{1}{2^{N-1}}\sum_{u_i=0,1}H_{\sum_{i=1}^{N-2}(-1)^{u_i}s_i}(s_{N-1})\\
:&=1-S_{N}+S_{N-1}.
\end{split}
\end{equation}
Then by \eqref{A1.2}, we have
\begin{equation}
\begin{split}
\mathcal{Q}_{A_1,A_2,\ldots,A_N}(\rho) = \sum_i\lambda_i\log_2\lambda_i+N-\frac{1}{2^{N}}\sum_{u_i=0,1}H_{\sum_{i=1}^{N-1}(-1)^{u_i}s_i}(|s_{N}|).
\end{split}
\end{equation}
\end{proof}

\section*{Appendix B: Quantum Discord for Multi-Partite GHZ state}
	\setcounter{equation}{0} \renewcommand\theequation{B\arabic{equation}}
Next, we will discuss the quantum discord of the N-partite GHZ state $\rho_{GHZ} = \mu|GHZ\rangle\langle GHZ|+\frac{1-\mu}{2^N}\mathbb{I}$, where $|GHZ\rangle = (|0\cdots 0\rangle+|1\cdots 1\rangle)/\sqrt{2}$. The state $\rho_{GHZ}$ in terms of Pauli matrices is given by
{\small
\begin{equation}
\begin{split}\label{BGHZPauliB}
\rho_{GHZ}=&\frac{1}{2^N}\mathbb{I}+\frac{\mu}{2^N}\sum_{t=1}^{\lfloor \frac{N}{2}\rfloor}\bigg[\sum_\pi \mathbf{S}_{\mathbf{\pi}}\Big(\sigma_3^{(1)}\otimes\cdots\otimes\sigma_3^{(2t)}\otimes\mathbb{I}^{ (2t+1)}\otimes\cdots\otimes\mathbb{I}^{(N)}\Big)\\
&+(-1)^t\sum_\pi \mathbf{S}_{\mathbf{\pi}}\Big(\sigma_2^{ (1)}\otimes\cdots\otimes\sigma_2^{(2t)}\otimes\sigma_1^{ (2t+1)}\otimes\cdots\otimes\sigma_1^{ (N)}\Big)\bigg]+\frac{\mu}{2^N}\sigma_1^{\otimes N}\\
=&\frac{1}{2^N}\mathbb{I}+\frac{\mu}{2^N}\sigma_1^{\otimes N}+\frac{\mu}{2^N}\sum_{t=1}^{\lfloor \frac{N}{2}\rfloor}\bigg[\sum_{\mathbf{\pi}} \mathbf{S}_{\pi}\Big(\mathbb{I}^{\otimes (N-2t)}\otimes\sigma_3^{\otimes 2t}\Big)\\
&+(-1)^t\sum_\pi \mathbf{S}_{\mathbf{\pi}}\Big(\sigma_1^{\otimes (N-2t)}\otimes\sigma_2^{\otimes 2t}\Big)\bigg],
\end{split}
\end{equation}}
where $\lfloor ~ \rfloor $ denotes the floor function, and $\mathbf{S}_{\mathbf{\pi}}$ permutes the tensor product by the permutation $\mathbf{\pi}\equiv \begin{pmatrix}
	1&2&\cdots N\\
	\pi_1&\pi_2&\cdots\pi_N
\end{pmatrix}$
 with $\pi_m\in \{1,\ldots,N\}$, i.e, the permutation $\mathbf{\pi}$ sends $v_1\otimes v_2\cdots\otimes v_N$ to $v_{\pi_1}\otimes v_{\pi_2}\cdots\otimes v_{\pi_N}$. $\sum_{\mathbf{\pi}}$ stands for the summation of all distinct permutations of the qubits.

\begin{theorem}
For the $N$-partite GHZ state in \eqref{BGHZPauliB}, the quantum discord is given by
{\small
\begin{equation}
\begin{split}\label{BQGHZ}
\mathcal{Q}_{A_1,\ldots,A_N}(\rho_{GHZ})&=\frac{1}{2^N}(1-\mu)\log_2(1-\mu)+\frac{1+(2^N-1)\mu}{2^N}\log_2(1+(2^N-1)\mu)\\
&-\frac{1}{2^{N-1}}\Big(1+(2^{N-1}-1)\mu\Big)\log_2\Big(1+(2^{N-1}-1)\mu\Big).
\end{split}
\end{equation}}
\end{theorem}

\begin{proof}
The proof is similar to 3-qubits and 4-qubits.
First, we calculate that the eigenvalues of $\rho_{GHZ}$:
\begin{equation}
\begin{split}\lambda_{1,\ldots,2^N-1}=\frac{1-\mu}{2^N}, \quad \lambda_{2^N}=\frac{1+(2^N-1)\mu}{2^N},
\end{split}
\end{equation}
then
{\small
\begin{equation}
\begin{split}S(\rho_{GHZ})=N-\frac{(1-\mu)(2^N-1)}{2^N}\log_2(1-\mu)-\frac{1+(2^N-1)\mu}{2^N}\log_2[1+(2^N-1)\mu].
\end{split}
\end{equation}}
Since $\rho_{A_1} = \mathrm{tr}_{A_2\cdots A_N}(\rho_{GHZ})=\mathbb{I}/2$, the entropy is calculated as $S_{A_1} = 1$.
After the $m$-th measurement ($m=1,\ldots,N-2$), the state is given by
$$
\rho_{\Pi^{A_1A_2\cdots A_m}} = \sum_{i=1}^{m}\sum_{u_i=0,1}p_{\vec{u}}\rho_{\vec{u}}
$$
with $\vec{u} = (u_1,\ldots,u_m)$. Here, each $u_i $ represents the measurement outcome of the subsystem $A_i (i= 0, 1)$.
Thus, the state on the subsystem $A_{m+1}$  is given by
\begin{equation}
\begin{split}
\mathrm{tr}_{A_{m+1}\cdots A_N}(\rho_{\vec{u}})=&\frac{1}{2(1+a_{\vec{u}})}V_{A_1}\Pi_{u_1}V_{A_1}^{\dagger}\otimes \cdots \otimes V_{A_m}^{(u_1\cdots u_{m-1})}\Pi_{u_m}\Big(V_{A_m}^{(u_1\cdots u_{m-1})}\Big)^{\dagger}\\
&\otimes(\mathbb{I}+a_{\vec{u}}\mathbb{I}+b_{\vec{u}}\sigma_3).
\end{split}
\end{equation}
where
{\small\begin{equation}
\begin{split}
a_{\vec{u}}&=\mu\sum_{j=1}^{\lfloor\frac{\widetilde{m}}{2}\rfloor}\sum_{\mathbf{\pi}}\mathbf{S}_{\mathbf{\pi}}
\Big((-1)^{u_{(1)}+\cdots+u_{(2j)}}1^{u_{(2j+1)}+\cdots+u_{(m)}}z_3^{(\mathbf{1})}\cdots z_3^{(\mathbf{2j})}\Big),\\
b_{\vec{u}}& = \mu\sum_{j=0}^{\lfloor\frac{\bar{m}}{2}\rfloor}\sum_{\mathbf{\pi}}\mathbf{S}_{\mathbf{\pi}}\Big((-1)^{u_{(1)}+\cdots
+u_{(2j+1)}}1^{u_{(2j+2)}+\cdots+u_{(m)}}z_3^{(\mathbf{1})}\cdots z_3^{(\mathbf{2j+1})}\Big)\end{split}
\end{equation}
}
with
$$
\widetilde{m}=
\begin{cases}
m-1 & \ \text{if} ~~~ m ~~~\text{odd} \\
m & \ \text{if}~~~ m ~~~\text{even}
\end{cases},~~~~~~~
\bar{m}=
\begin{cases}
m & \ \text{if} ~~~ m ~~~\text{odd} \\
m-1 & \ \text{if}~~~ m ~~~\text{even}
\end{cases}
$$
and for $i=1,2,3$
$$z_i^{(\mathbf{j})}:=
\begin{cases}
z_i & \ \text{if} ~~~ \mathbf{j}=1, \\
z_i^{(u_1u_2\cdots u_{j-1})} & \text{otherwise}
\end{cases}$$
satisfying
$(z_1^{(\mathbf{j})})^2+(z_2^{(\mathbf{j})})^2+(z_3^{(\mathbf{j})})^2 = 1$ for all $\mathbf{j}$. Note that each different selection of $u_1,\ldots,u_j$ corresponds to a set of variables.
For example: if $\mathbf{j}=2$, then $z_i^{(\mathbf{2})}$ corresponds to $z_1^{(0)}$, $z_2^{(0)}$, $z_3^{(0)}$, $z_1^{(1)}$, $z_2^{(1)}$, $z_3^{(1)}$.
The variables involved increase exponentially with the number of measurements.
The entropy is given by
\begin{equation}
\begin{split}
S_{A_{m+1}|\Pi^{A_1\cdots A_m}}(\rho_{GHZ}) = 1-\frac{1}{2^{m+1}}\sum_{i=1}^{m}\sum_{u_i=0,1}H_{a_{\vec{u}}}(b_{\vec{u}})+\frac{1}{2^m}\sum_{i=1}^{m}\sum_{u_i=0,1}H(a_{\vec{u}}).
\end{split}
\end{equation}
After the $(N-1)$ measurement, the state will change to the ensemble $\{\rho_{_{\vec{u}}}, p_{_{\vec{u}}}\}$ with
\begin{equation}
\begin{split}
\rho_{_{\vec{u}}}=&\frac{1}{2(1+\alpha_{\vec{u}})}V_{A_1}\Pi_{u_1}V_{A_1}^{\dagger}\otimes \cdots \otimes V_{A_{N-1}}^{(u_1\cdots u_{N-2})}\Pi_{u_{N-1}}\Big(V_{A_{N-1}}^{(u_1\cdots u_{{N-2}})}\Big)^{\dagger}\\
&\otimes\left(\mathbb{I}+\alpha_{\vec{u}} \mathbb{I}+\beta_{\vec{u}}\sigma_1+\gamma_{\vec{u}}\sigma_2+\delta_{\vec{u}}\sigma_3\right).
\end{split}
\end{equation}
Here, $\vec{u} = (u_1,\ldots,u_{N-1})$ and
{\scriptsize
\begin{equation}
\begin{split}
\alpha_{\vec{u}} &= \mu\sum_{t=1}^{\lfloor\frac{\bar{N}-1}{2}\rfloor}\sum_{\mathbf{\pi}}\mathbf{S}_{\mathbf{\pi}}\Big((-1)^{u_{(1)}+\cdots+u_{(2t)}}
1^{u_{(2t+1)}+\cdots+u_{(N-1)}}z_3^{(\mathbf{1})}\cdots z_3^{(\mathbf{2t})}\Big),\\
\beta_{\vec{u}} &= \mu\sum_{t=1}^{\lfloor\frac{\bar{N}-1}{2}\rfloor}\sum_{\mathbf{\pi}}\mathbf{S}_{\mathbf{\pi}}\Big((-1)^{u_{(1)}+\cdots+u_{(N-1)}+t}
z_2^{(\mathbf{1})}\cdots z_2^{(\mathbf{2t})}z_1^{(\mathbf{2t+1})}\cdots z_1^{(\mathbf{N-1})}\Big),\\
\gamma_{\vec{u}} &= \mu\sum_{t=0}^{\lfloor\frac{\widetilde{N}-1}{2}\rfloor}\sum_{\mathbf{\pi}}\mathbf{S}_{\mathbf{\pi}}\Big((-1)^{u_{(1)}+\cdots+u_{(N-1)}+t}
z_2^{(\mathbf{1})}\cdots z_2^{(\mathbf{2t+1})}z_1^{(\mathbf{2t+2})}\cdots z_1^{(\mathbf{N-1})}\Big),\\
\delta_{\vec{u}} &= \mu\sum_{t=0}^{\lfloor\frac{\widetilde{N}-1}{2}\rfloor}\sum_{\mathbf{\pi}}\mathbf{S}_{\mathbf{\pi}}\Big((-1)^{u_{(1)}+\cdots+u_{(2t+1)}}
1^{u_{(2t+2)}+\cdots+u_{(N-1)}}z_3^{(\mathbf{1})}\cdots z_3^{(\mathbf{2t+1})}\Big).\end{split}
\end{equation}
}
Similar to the previous proof,
\begin{equation}
\begin{split}\min_{\Pi^{A_1\ldots A_{N-1}}}&\{S_{A_2|\Pi^{A_1}}(\rho)+\cdots+S_{A_{N}|\Pi^{A_1\ldots A_{N-1}}}(\rho)\}\\
&=\min{S_{A_2|\Pi^{A_1}}(\rho)}+\cdots+\min{S_{A_{N}|\Pi^{A_1\ldots A_{N-1}}}(\rho)}
\end{split}
\end{equation}
and
{\small
\begin{equation}
\begin{split}
S(\rho_{A_1})&=1;\\
\min{S_{A_2|\Pi^{A_1}}(\rho)}&=1-\frac{1}{2}H(\mu z_3):=1-S_2;\\
&\cdots\\
\min{S_{A_{N-1}|\Pi^{A_1\ldots A_{N-2}}}(\rho) } &=1-\frac{1}{2^{N-1}}\sum_{i=1}^{N-2}\sum_{u_i=0,1}H_{\bar{a}_{\vec{u}}}(\bar{b}_{\vec{u}})+\frac{1}{2^{N-2}}\sum_{i=1}^{N-3}\sum_{u_i=0,1}H(\bar{a}_{\vec{u}})\\
&=:1-S_{N-1}+S_{N-2};\\
\min{S_{A_{N}|\Pi^{A_1\ldots A_{N-1}}}(\rho)} &= 1-\frac{1}{2^{N}}\sum_{i=1}^{N-1}\sum_{u_i=0,1}H_{\bar{\alpha}_{\vec{u}}}(\bar{\delta}_{\vec{u}})+\frac{1}{2^{N-1}}\sum_{i=1}^{N-2}\sum_{u_i=0,1}H(\bar{\alpha}_{\vec{u}})\\
& =: 1-S_{N}+S_{N-1},
\end{split}
\end{equation}}
where \begin{equation}
\begin{split}\bar{a}_{\vec{u}} &= \mu\sum_{t=1}^{\lfloor\frac{\bar{N}-1}{2}\rfloor}\sum_{\pi}\mathbf{S}_{\pi}\Big((-1)^{u_{(1)}+\cdots+u_{(2t)}}1^{u_{(2t+1)}+\cdots+u_{(N-1)}}\Big),\\
\bar{b}_{\vec{u}} &= \mu\sum_{t=0}^{\lfloor\frac{\widetilde{N}-1}{2}\rfloor}\sum_{\pi}\mathbf{S}_{\pi}\Big((-1)^{u_{(1)}+\cdots+u_{(2t+1)}}1^{u_{(2t+2)}+\cdots+u_{(N-1)}}\Big),\\
\bar{\alpha}_{\vec{u}} &= \mu\sum_{t=1}^{\lfloor\frac{\bar{N}-1}{2}\rfloor}\sum_{\pi}\mathbf{S}_{\pi}\Big((-1)^{u_{(1)}+u_{(2)}+\cdots+u_{(2t)}}1^{u_{(2t+1)}+\cdots+u_{(N-1)}}\Big),\\
\bar{\delta}_{\vec{u}} &= \mu\sum_{t=0}^{\lfloor\frac{\widetilde{N}-1}{2}\rfloor}\sum_{\pi}\mathbf{S}_{\pi}\Big((-1)^{u_{(1)}+u_{(2)}+\cdots+u_{(2t+1)}}1^{u_{(2t+2)}+\cdots+u_{(N-1)}}\Big).\end{split}
\end{equation}
Due to symmetry, we obtain the following classification for $\bar{\alpha}_{\vec{u}}$ and $\bar{\delta}_{\vec{u}}$:
\begin{equation}
\begin{split}
\bar{\alpha}_{\vec{u}}&=
\begin{cases}
\mu\sum_{t=1}^{\lfloor\frac{\widetilde{N}-1}{2}\rfloor}\tbinom{N-1}{2t} & \ \text{if all the elements} ~~u_i~~ \text{in} ~~\vec{u}~~ \text{are all 0 or all 1},\\
-\mu & \ \text{otherwise},
\end{cases}\\
\bar{\delta}_{\vec{u}}&=
\begin{cases}
\mu\sum_{t=0}^{\lfloor\frac{\widetilde{N}-1}{2}\rfloor}\tbinom{N-1}{2t+1} & \ \text{if all the elements} ~~u_i~~ \text{in} ~~\vec{u}~~ \text{are all 0} ,\\
-\mu\sum_{t=0}^{\lfloor\frac{\widetilde{N}-1}{2}\rfloor}\tbinom{N-1}{2t+1} & \ \text{if all the elements} ~~u_i~~ \text{in} ~~\vec{u}~~ \text{are all 1},\\
0 & \ \text{otherwise}.
\end{cases}
\end{split}
\end{equation}
Thus,
{\small
\begin{equation}
\begin{split}&\min_{\Pi^{A_1\ldots A_{N-1}}}\{S_{A_2|\Pi^{A_1}}(\rho)+\cdots+S_{A_{N}|\Pi^{A_1\ldots A_{N-1}}}(\rho)\}=N-\frac{1}{2^{N}}\sum_{i=1}^{N-1}\sum_{u_i=0,1}H_{\bar{\alpha}_{\vec{u}}}(\bar{\delta}_{\vec{u}})\\
&=N-\frac{2^{N-1}-1}{2^{N-1}}(1-\mu)\log_2(1-\mu)-\frac{1+(2^{N-1}-1)\mu}{2^{N-1}}\log_2\Big(1+(2^{N-1}-1)\mu\Big).
\end{split}
\end{equation}}
Then \eqref{A1.2} implies the result \eqref{BQGHZ}.
\end{proof}
\vskip 0.1in

\bibliographystyle{amsalpha}